\documentclass[aps,nopacs,nokeys,superscriptaddress,11pt,twoside,notitlepage]{revtex4-1}

\usepackage{graphicx,epic,eepic,epsfig,amsmath,latexsym,amssymb,verbatim,color}

\usepackage{theorem}
\newtheorem{definition}{Definition}
\newtheorem{proposition}[definition]{Proposition}
\newtheorem{lemma}[definition]{Lemma}

\newtheorem{theorem}[definition]{Theorem}
\newtheorem{corollary}[definition]{Corollary}

\def\squareforqed{\hbox{\rlap{$\sqcap$}$\sqcup$}}
\def\qed{\ifmmode\squareforqed\else{\unskip\nobreak\hfil
\penalty50\hskip1em\null\nobreak\hfil\squareforqed
\parfillskip=0pt\finalhyphendemerits=0\endgraf}\fi}
\def\endenv{\ifmmode\;\else{\unskip\nobreak\hfil
\penalty50\hskip1em\null\nobreak\hfil\;
\parfillskip=0pt\finalhyphendemerits=0\endgraf}\fi}
\newenvironment{proof}{\noindent \textbf{{Proof~} }}{\qed}
\newenvironment{remark}{\noindent \textbf{{Remark~}}}{\qed}

\newenvironment{biography}[1]{\noindent \textbf{{#1}~}}{}

\newenvironment{beweis}[1]{\noindent\textbf{Proof~{#1}~} }{\qed}

\mathchardef\ordinarycolon\mathcode`\:
\mathcode`\:=\string"8000
\def\vcentcolon{\mathrel{\mathop\ordinarycolon}}
\begingroup \catcode`\:=\active
  \lowercase{\endgroup
  \let :\vcentcolon
  }

\newcommand{\nc}{\newcommand}
\nc{\rnc}{\renewcommand}
\nc{\beg}{\begin{equation}}
\nc{\eeq}{{\end{equation}}}
\nc{\beqa}{\begin{eqnarray}}
\nc{\eeqa}{\end{eqnarray}}
\nc{\lbar}[1]{\overline{#1}}
\nc{\bra}[1]{\langle#1|}
\nc{\ket}[1]{|#1\rangle}
\nc{\ketbra}[2]{|#1\rangle\!\langle#2|}
\nc{\braket}[2]{\langle#1|#2\rangle}

\nc{\proj}[1]{| #1\rangle\!\langle #1 |}
\nc{\avg}[1]{\langle#1\rangle}
\nc{\Rank}{\operatorname{Rank}}
\nc{\smfrac}[2]{\mbox{$\frac{#1}{#2}$}}
\nc{\tr}{\operatorname{Tr}}
\nc{\ox}{\otimes}
\nc{\dg}{\dagger}
\nc{\dn}{\downarrow}
\nc{\cA}{{\cal A}}
\nc{\cB}{{\cal B}}
\nc{\cC}{{\cal C}}
\nc{\cD}{{\cal D}}
\nc{\cE}{{\cal E}}
\nc{\cF}{{\cal F}}
\nc{\cG}{{\cal G}}
\nc{\cH}{{\cal H}}
\nc{\cI}{{\cal I}}
\nc{\cJ}{{\cal J}}
\nc{\cK}{{\cal K}}
\nc{\cL}{{\cal L}}
\nc{\cM}{{\cal M}}
\nc{\cN}{{\cal N}}
\nc{\cO}{{\cal O}}
\nc{\cP}{{\cal P}}
\nc{\cQ}{{\cal Q}}
\nc{\cR}{{\cal R}}
\nc{\cS}{{\cal S}}
\nc{\cT}{{\cal T}}
\nc{\cX}{{\cal X}}
\nc{\cZ}{{\cal Z}}
\nc{\csupp}{{\operatorname{csupp}}}
\nc{\qsupp}{{\operatorname{qsupp}}}
\nc{\var}{{\operatorname{var}}}
\nc{\rar}{\rightarrow}
\nc{\lrar}{\longrightarrow}
\nc{\polylog}{{\operatorname{polylog}}}
\nc{\1}{{\openone}}
\nc{\wt}{{\operatorname{wt}}}
\nc{\av}[1]{{\left\langle {#1} \right\rangle}}

\def\a{\alpha}
\def\b{\beta}

\def\d{\delta}

\def\ll{\lambda}

\def\G{\Gamma}

\def\S{\Sigma}
\def\U{\Upsilon}

\def\O{\Omega}

\nc{\RR}{{{\mathbb R}}}
\nc{\CC}{{{\mathbb C}}}
\nc{\FF}{{{\mathbb F}}}
\nc{\NN}{{{\mathbb N}}}
\nc{\ZZ}{{{\mathbb Z}}}
\nc{\PP}{{{\mathbb P}}}
\nc{\QQ}{{{\mathbb Q}}}
\nc{\UU}{{{\mathbb U}}}
\nc{\EE}{{{\mathbb E}}}
\nc{\id}{{\operatorname{id}}}

\nc{\CHSH}{{\operatorname{CHSH}}}

\nc{\be}{\begin{equation}}
\nc{\ee}{{\end{equation}}}
\nc{\bea}{\begin{eqnarray}}
\nc{\eea}{\end{eqnarray}}
\nc{\<}{\langle}
\rnc{\>}{\rangle}
\nc{\Hom}[2]{\mbox{Hom}(\CC^{#1},\CC^{#2})}
\nc{\rU}{\mbox{U}}

\nc{\ob}[1]{#1}

\nc{\SEP}{{\text{SEP}}}
\nc{\NS}{{\text{NS}}}
\nc{\LOCC}{{\text{LOCC}}}
\nc{\PPT}{{\text{PPT}}}
\nc{\EXT}{{\text{EXT}}}
\nc{\Sym}{{\operatorname{Sym}}}

\nc{\ERLO}{{E_{\text{r,LO}}}}
\nc{\ERLOCC}{{E_{\text{r,LOCC}}}}
\nc{\ERPPT}{{E_{\text{r,PPT}}}}
\nc{\ERLOCCinfty}{{E^{\infty}_{\text{r,LOCC}}}}
\nc{\Aram}{{\operatorname{\sf A}}}

\begin{document}

\title{No-Signalling Assisted Zero-Error Capacity of Quantum Channels\protect\\
and an\protect\\ Information Theoretic Interpretation of the Lov\'{a}sz Number}

\author{Runyao Duan}

\email{runyao.duan@uts.edu.au}

\affiliation{Centre for Quantum Computation and Intelligent Systems (QCIS), Faculty of Engineering and Information Technology, University of Technology, Sydney, NSW 2007, Australia}
\affiliation{State Key Laboratory of Intelligent Technology and Systems, Tsinghua National Laboratory for Information Science and Technology, Department of Computer Science and Technology, Tsinghua University, Beijing 100084, China}
\affiliation{UTS-AMSS Joint Research Laboratory for Quantum Computation and Quantum Information Processing, Academy of Mathematics and Systems Science, Chinese Academy of Sciences, Beijing 100190, China}

\author{Andreas Winter}

\email{andreas.winter@uab.cat}

\affiliation{ICREA \&{} F\'{\i}sica Te\`orica: Informaci\'o i Fen\`{o}mens Qu\`antics, Universitat Aut\`onoma de Barcelona, ES-08193 Bellaterra (Barcelona), Spain}
\affiliation{School of Mathematics, University of Bristol, Bristol BS8 1TW, United Kingdom}

\date{30 November 2015}

\begin{abstract}
\vspace{1cm}
We study the one-shot zero-error classical capacity of 
a quantum channel assisted by quantum no-signalling correlations, and the 
reverse problem of exact simulation of a prescribed channel by a noiseless
classical one.
Quantum no-signalling correlations are viewed as two-input and two-output
completely positive and trace preserving maps with linear constraints enforcing
that the device cannot signal.
Both problems lead to simple semidefinite programmes (SDPs) that
depend only on the Kraus operator space of the channel. In particular, we 
show that the zero-error classical simulation cost is precisely 
the conditional min-entropy of the Choi-Jamio\l{}kowski matrix of the given 
channel. The zero-error classical capacity is given by a similar-looking
but different SDP; the asymptotic zero-error classical capacity
is the regularization of this SDP, and in general we do not know of
any simple form.

Interestingly however, for the class of classical-quantum channels, we show that the 
asymptotic capacity is given by a much simpler SDP, 
which coincides with a semidefinite generalization of the fractional packing number 
suggested earlier by Aram Harrow. 
This finally results in an operational interpretation 
of the celebrated Lov\'asz $\vartheta$ function of a graph as the 
zero-error classical capacity of the graph assisted by quantum no-signalling 
correlations, the first information theoretic interpretation of the
Lov\'{a}sz number. 
\end{abstract}

\maketitle

\thispagestyle{empty}

\vfill\pagebreak

\tableofcontents


\setcounter{page}{1}

\section{Introduction} 
\label{sec:intro}
We choose as the starting point of the present 
work the fundamental problem of channel simulation. 
Roughly speaking, this problem asks when a communication channel 
$\cN$ from Alice (A) to Bob (B) can be used to simulate another channel 
$\cM$, also from A to B? \cite{KretschmannWerner:tema} 
This problem has many variants according to the resources available to A and B. 
In particular, the case when A and B can access unlimited amount of shared 
entanglement has been completely solved. Let $C_E(\cN)$ denote the entanglement-assisted classical capacity of $\cN$ \cite{BSST2003}.  
It was shown that, in the asymptotic setting, to optimally simulate $\cM$, 
we need to apply $\cN$ at rate $\frac{C_E(\cM)}{C_E(\cN)}$ \cite{BDHS+2009,QRST-simple}. 
In other words, the entanglement-assisted classical capacity uniquely 
determines the properties of the channel in the simulation process. 
Furthermore, even with stronger resources such as no-signalling correlations or
feedback, this rate cannot be improved -- otherwise we would violate causality,
see \cite{BDHS+2009} for a discussion. 

Here we are interested in the zero-error setting \cite{Shannon1956}. 
It is well known that the zero-error 
communication problem is extremely difficult, already for classical channels. 
Indeed, the single-shot zero-error classical communication capability of a 
classical noisy channel equals the independence number of the
(classical) confusability graph induced by the channel, 
and the latter problem is well-known to be NP-complete. 
The behaviour of quantum channels in zero-error communication is 
even more complex as striking effects such as super-activation are 
possible \cite{DS2008, Duan2009, CCH2009, CS2012}. 
The most general zero-error simulation problem remains wide open. 
To overcome this difficulty, many variants of this problem have been proposed. 
The most natural way is to introduce some additional resources and see 
how this changes the capacity. Indeed, extra resources such as classical 
feedback \cite{Shannon1956}, 
entanglement \cite{Duan2009, CLMW2010, DSW2010}, and even a small (constant) 
amount of forward communication \cite{CLMW2011}, have been introduced. It has been 
shown these extra resources can increase the capacity, and 
generally simplify the problem. In particular, it was shown that even 
for classical communication channel, shared entanglement can strictly increase 
the asymptotic zero-error classical capacity \cite{LMMO+2012}. However, determining
the entanglement-assisted zero-error classical capacity remains an open problem 
even for classical channels. More powerful resources are actually required in 
order to simplify the problem. Cubitt \textit{et al.} \cite{CLMW2011}
introduced classical no-signalling correlations into the zero-error communication 
for classical channels, and showed that the well-known fractional packing 
number of the bipartite graph induced by the channel, gives precisely the 
zero-error classical capacity of the channel. Previously, it was known by 
Shannon that this fractional packing number corresponds to the zero-error 
classical capacity of the channel when assisted with a feedback link from 
the receiver to the sender and when the unassisted zero-error classical 
capacity is not vanishing \cite{Shannon1956}. For general background on
graph theory see \cite{Berge}, and for ``fractional graph theory'' the
delightful book \cite{ScheinermanUllman}.

Another major motivation for this work is to further explore the connection 
between quantum information theory and the so-called ``non-commutative graph 
theory'' suggested in \cite{DSW2010}. Such a connection has been well-known 
in classical information theory. In \cite{Shannon1956}, Shannon realized 
that the zero-error capacity of a classical noisy channel only depends 
on the confusability graph induced by the channel. He further pointed out 
that in the presence of classical feedback, the zero-error capacity is 
completely determined by the bipartite graph of possible input-output 
transitions associated with the channel. 
Thus it makes sense to talk about the zero-error capacity of a (bipartite) 
graph. The notion of non-commutative graph naturally occurs when we use 
quantum channels for zero-error communication. For any quantum channel, 
the non-commutative graph associated with the channel captures 
the zero-error communication properties, thus playing a similar role to 
confusability graph. Most notably, this notion also makes it possible to 
introduce a quantum Lov\'asz $\vartheta$ function to upper bound the 
entanglement-assisted zero-error capacity that has properties quite similar 
to its classical analogue \cite{DSW2010}. Very recently, it was shown 
that the zero-error classical capacity of a quantum channel in the presence 
of quantum feedback only depends on the Kraus operator space of the 
channel \cite{DSW2013}. In other words, the Kraus operator space plays a 
role that is quite similar to the bipartite graph. Now it becomes clear 
that any classical channel induces a bipartite graph as well as a confusability 
graph, while a quantum channel induces a non-commutative bipartite graph 
and a non-commutative graph. The new insight is that we can simply regard 
a non-commutative (bipartite) graph as a high-level abstraction of all 
underlying quantum channels, and study its information-theoretic properties, 
not limited to zero-error setting. This leads us to a very general viewpoint: 
graphs as communication channels. For instance, we can define the 
entanglement-assisted classical capacity of a non-commutative bipartite 
graph as the minimum of the entanglement-assisted classical capacity of 
quantum channels that induce the given Kraus operator space. It was shown 
that this quantity enjoys a number of interesting properties including 
additivity under tensor product  and an operational interpretation as a 
sort of entanglement-assisted conclusive capacity of the bipartite 
graph \cite{DSW2013}. It remains a great challenge to find tractable forms 
of various capacities for non-commutative (bipartite) graphs. 

In this paper we consider a more general class of quantum no-signalling 
correlations described by two-input and two-output quantum channels with 
the no-signalling constraints. This kind of correlations naturally arises 
in the study of the relativistic causality of quantum 
operations \cite{BGNP2001,ESW2001,PHHH2006}; see also the more
recent \cite{OCB2012}.
Distinguishability of these correlations from an information theoretic 
viewpoint has also been studied \cite{Chiribella2012}. We provide a 
number of new properties of these correlations, and establish several 
structural theorems of these correlations. Then we generalize the approach 
of \cite{CLMW2011} to study the zero-error classical capacity of a 
noisy quantum channel assisted by quantum no-signalling correlations, 
and the reverse problem of perfect simulation. 
We show that both problems can be completely solved in the one-shot scenario, 
revealing some nice structure:
\begin{enumerate}
  \item The answers are given by semidefinite programmes (SDPs, cf.~\cite{SDP});
  
  \item At the same time they generalize the results of Cubitt \emph{et al.} \cite{CLMW2011}; 

  \item For the simulation, the question is really how to form a constant 
  channel by a convex combination of the one we want to simulate and an
  arbitrary other quantum channel, and the number of bits needed is just 
  $-\log p$, where $p$ is the probability weight of the target channel 
  in the convex combination (throughout this paper, $\log$ denotes the
  binary logarithm); 
  
  \item For assisted communication, there is an analogous problem of 
  convex-combing a certain channel from B to A which has some kind of 
  orthogonality relation with the given channel from A to B, with 
  another one to form a constant channel. If the target channel has 
  weight $p$, then the number of bits sent is again $-\log p$. 
\end{enumerate}

Most interestingly, the solution to the communication problem only 
depends on the Kraus operator space of the channel, not directly on
the channel itself. For the simulation problem, the solution is given 
by the conditional min-entropy \cite{KRS2009,Tomamichel-PhD} 
of the channel's Choi-Jamio\l{}kowski matrix, 
and is actually additive, thus also gives the asymptotic cost of 
simulating the channel. If we are interested in simulating the cheapest 
channel contained in the Kraus operator space, we obtain an SDP in terms 
of the projection of the Choi-Jamio\l{}kowski matrix. Both the
capacity and the simulation SDPs are in general not known to 
be multiplicative under the tensor product of channels, 
thus we do not know the optimal asymptotic simulation cost.

We then focus on the asymptotic zero-error classical capacity and simulation 
cost assisted with quantum no-signalling correlations. This requires 
determining the asymptotic behaviour of a sequence of SDPs. In general 
the one-shot solution does not give the asymptotic result, since the 
corresponding SDP is not multiplicative with respect to the tensor product 
of channels.
A simple formula for the asymptotic channel capacity remains unknown. 
However, for the special cases of classical-quantum (cq) channels,
we find that the zero-error capacity is given by the solution of a 
rather simple SDP suggested earlier by Harrow as a natural generalization 
of the classical fractional packing number \cite{Harrow2010}, which we call
\emph{semidefinite packing number}. This result has two interesting 
corollaries. First, it implies that the zero-error classical capacity 
of cq-channels assisted by quantum no-signalling correlations is additive. 
Second, and more importantly, we show that for a classical graph $G$, 
the celebrated Lov\'asz number $\vartheta(G)$~\cite{Lovasz1979}, is actually 
the minimum zero-error classical capacity of any 
cq-channel that has the given graph as its confusability graph. In other 
words, Lov\'asz' $\vartheta$ function is the zero-error classical capacity 
of a graph assisted by quantum no-signalling correlations. To the best 
of our knowledge, this is the first information theoretic 
operational interpretation  of the Lov\'asz number since its introduction
in 1979. Previously, it was known  that it is an upper bound on the 
entanglement-assisted zero-error classical capacity of a 
graph \cite{Beigi2010,DSW2010}. 
It remains unknown whether the use of quantum no-signalling correlations 
could be replaced by shared entanglement. The asymptotic simulation 
cost for Kraus operator spaces associated with cq-channels is rather 
simpler, and is actually given by the one-shot simulation cost. 

\medskip
Before we proceed to the technical details, it may be helpful to present an 
overview of our main results. Let $\cN$ be a quantum channel from $\cL(A')$ 
to $\cL(B)$, with a Kraus operator sum representation 
$\cN(\rho)=\sum_k E_k\rho E_k^\dag,$ where $\sum_k E_k^\dag E_k=\1_{A'}$. 
Let $K=K(\cN)=\operatorname{span}\{E_k\}$ denote the Kraus operator space of $\cN$. 
The Choi-Jamio\l{}kowski matrix of $\cN$ is given by 
$J_{AB}=\sum_{ij} \ketbra{i}{j}_{A}\ox \cN(\ketbra{i}{j}_{A'})=(\id_A\ox \cN)(\ketbra{\Phi_{AA'}}{\Phi_{AA'}})$, where
$A$ and $A'$ are isomorphic Hilbert spaces, $\{\ket{i}\}$ 
($\{\ket{j}\}$) is orthonormal basis over $A$ ($A'$, resp.), and $\ket{\Phi_{AA'}}=\sum_{k}\ket{k_{A}}\ket{k_{A'}}$ is 
the unnormalized maximally entangled state over $A\ox A'$. Recall that $\tr_B J_{AB}=\1_A$.
Let $P_{AB}$ denote the projection onto the support of $J_{AB}$, which
is the subspace $(\1_A\ox K)\ket{\Phi_{AA'}}$, sometimes called the Choi-Jamio\l{}kowski support of $K$ (or of the channel). 

It is worth noting that many results below can be defined on any matrix subspace 
$K$, not just those corresponding to a quantum channel $\cN$. However, we have 
to make sure that $K$ is actually corresponding to some quantum channel $\cN$. 
This puts an additional constraint on $K$. 
More precisely, suppose $K=\operatorname{span}\{E_k\}$ for some orthonormal basis 
$\{E_k\}$ such that $\tr E_j^\dag E_k=\delta_{jk}$. Then we should be able to 
find a quantum channel $\cN'=\sum_j F_j\cdot F_j^\dag$ such that  
$F_{j}=\sum_k a_{jk}E_k$, for some invertible matrix $a=[a_{jk}]$ and 
$\sum_{j}F_j^\dagger F_j=\1_{A'}$. This is equivalent to 
\[
  \sum_{kl} r_{kl}E_{k}^\dagger E_{l} = \1_{A'}, \text{ where } r_{kl} = \sum_{j}a_{jk}^*a_{jl}.
\]
If such a positive definite matrix $r=[r_{kl}]$ cannot be found, $K$ 
will not correspond to a quantum channel
(nor any $K^{\ox n}$ for $n\ge1$).  
In this case one might still be able to find $K'\subset K$ 
that is a Kraus operator space for some quantum channel $\cN'$. For 
instance, $K={\rm span}\{\1, \ketbra{0}{1}\}$, and $K'={\rm span}\{\1\}$. 
(Note however that $K$ is the limit of the Kraus subspaces of genuine
quantum channels, namely amplitude damping channels with damping
parameter going to $0$; hence it might still be considered as
admissible Kraus space of an infinitesimally amplitude damping channel.)
We will, therefore, always assume that $K$ corresponds to some quantum channel 
$\cN$ such that $K=K(\cN)$. From now on, any such Kraus operator space $K$ 
will be alternatively called ``non-commutative bipartite graph'' -- in fact,
below we shall argue why it is a natural generalization of bipartite graphs. 

\begin{theorem}
  \label{capacity-1-shot}
  The one-shot zero-error classical capability, (quantified as the largest
  number of messages), of 
  $\cN$ assisted by quantum no-signalling correlations depends only on
  the non-commutative graph $K$, and is given by the 
  integer part of the following SDP:
  \[\begin{split}
    \U(\cN) & =\U(K) \\
            & =\max \tr S_A\ \text{ s.t. }\ 0\leq E_{AB}\leq S_A\ox \1_B, 
                                         \tr_A E_{AB}=\1_B, \tr P_{AB}(S_A\ox \1_B-E_{AB})=0,
  \end{split}\]
  where $P_{AB}$ denotes the projection onto the the subspace $(\1\ox K)\ket{\Phi}$. \\
 Hence we are motivated to call $\U(K)$ the 
  \emph{no-signalling assisted independence number of $K$.}
\end{theorem}
The proof of this theorem will be given in Section \ref{subsec:1-shot-capacity}, 
where we also explore other properties of the above SDP.
For another direction of investigation, looking at the specific
channel $\cN$ with Choi-Jamio\l{}kowski matrix $J$ and trying to minimize 
the error probability for given number of messages, we refer the reader 
to the recent and highly relevant work of Leung and Matthews~\cite{LeungMatthews}, 
where exactly this is done in an environment with free no-signalling
resources subject to other semidefinite constraints.

It is evident from this theorem that the one-shot zero-error classical 
capacity of $\cN$ only depends on the Kraus operator space $K$. That is, any 
two quantum channels $\cN$ and $\cM$ will have the same capacity if they 
have the same Kraus operator space. For this reason, we usually use $\U(K)$ 
to denote $\U(\cN)$. Furthermore, we can talk about the capacity of the 
Kraus operator space $K$ directly without referring to the underling quantum 
channel. Notice that a classical channel $N=(X, p(y|x), Y)$ naturally induces 
a bipartite graph $\G(N)=(X, E, Y)$, where the input and output alphabets
$X$ and $Y$ are the two sets of vertices, and 
$E\subset X\times Y$ is the set of edges such that $(x,y)\in E$ 
if and only if $p(y|x)>0$. We shall also use the notation $\Gamma(y|x)=1$
if $(x,y)\in E$, and $\Gamma(y|x)=0$ otherwise. In this case, we have
$$K(N)=K(\G) = \operatorname{span} \{ \ketbra{y}{x} : p(y|x)>0 \},$$
and our notion $K(\cN)$ generalizes this to arbitrary quantum channels.

Similarly, the simulation cost is given as follows. 
\begin{theorem} 
  \label{review-simulation-1-shot}
  The one-shot zero-error classical simulation cost (quantified as the minimum
  number of messages) of a quantum channel $\cN$ under quantum no-signalling assistance is given by 
  $\lceil 2^{-H_{\min}(A|B)_{J}}\rceil$. 
  Here, $J_{AB}=(\id_A\ox \cN)\Phi_{AA'}$ is the Choi-Jamio\l{}kowski  matrix of $\cN$, and  
$H_{\min}(A|B)_{J}$ is the conditional min-entropy 
  defined as follows \cite{KRS2009,Tomamichel-PhD}:
  \[
    2^{-H_{\min}(A|B)_{J}} = \S(\cN)= \min \tr T_B,\ {\rm s.t. }\ J_{AB}\leq \1_A\ox T_B.
  \]
\end{theorem}

For example, the asymptotic zero-error classical simulation 
cost of the cq-channel $0\rightarrow \rho_0$ and $1\rightarrow \rho_1$, 
is given by $\log(1+D(\rho_0,\rho_1)),$ where 
$D(\rho_0,\rho_1) = \frac12 \| \rho_0-\rho_1 \|_1$ is the trace distance 
between $\rho_0$ and $\rho_1$. This gives a new operational interpretation 
of the trace distance between $\rho_0$ and $\rho_1$ as the asymptotic exact 
simulation cost for the above cq-channel.

Since there might be more than one channel with Kraus operator space 
included in $K$, we are interested in the exact simulation cost of the 
``cheapest'' among these channels. More precisely, the one-shot zero-error 
classical simulation cost of a Kraus operator space $K$ is defined as 
$$\S(K)=\min \{\S(\cN): \cN \text{ is quantum channel and } K(\cN)<K\},$$
where $K(\cN)<K$ means that $K(\cN)$ is a subspace of $K$. 
Then it follows immediately from Theorem~\ref{review-simulation-1-shot} that 
\begin{theorem}
  \label{thm:L-1-shot-cost}
  The one-shot zero-error classical simulation cost of a Kraus operator space $K$ under quantum no-signalling assistance
  is given by the integer ceiling of  
  $$\S(K)=\min \tr T_B \text{ s.t. } 0\leq F_{AB}\leq \1_A\ox T_B,\ 
                                     \tr_B F_{AB}=\1_A,\ \tr F_{AB}(\1_{AB}-P_{AB})=0,$$
  where $P_{AB}$ is the projection onto the Choi-Jamio\l{}kowski support of $K$.
\end{theorem}
We will prove these two theorems in Section~\ref{subsec:1-shot-cost}.

\medskip
We introduce the asymptotic zero-error channel capacity of $K$ by considering 
the number of bits that can be communicated over $n$ copies of the channel $\cN$,
i.e.~$\cN^{\ox n}$, having Kraus operator space $K^{\ox n}$, per 
channel use as $n \rightarrow \infty$; we denote it as $C_{0,\NS}(K)$.
Likewise, the asymptotic number of bits needed per channel use to simulate
$\cN^{\ox n}$ as $n \rightarrow \infty$, denoted $S_{0,\NS}(\cN)$, and the 
same minimized over all channels with Kraus operator space $K^{\ox n}$
(not necessarily product channels!), which we denote $S_{0,\NS}(K)$.

From these definitions, it is clear that they are given by the regularizations
of the respective one-shot quantities:
\begin{align}
  \label{eq:channel-sim}
  S_{0,\NS}(\cN) &= -H_{\min}(A|B)_{J}, \\
  \label{eq:K-cap+sim}
  C_{0,\NS}(K)    = \sup_{n\geq 1} \frac1n \log \U\left(K^{\ox n}\right),
  \quad
  S_{0,\NS}(K)   &= \inf_{n\geq 1}\frac1n \log \S\left(K^{\ox n}\right),
\end{align}
the first one because $H_{\min}(A|B)$ is additive under tensor products.

So far, we are unable to determine closed formulas for the latter two in general. 
Interestingly, the special case of $K$ coming from a cq-channel can 
be solved completely. Note that if $K$ corresponds to a cq-channel 
$\cN: i \longmapsto \rho_i$, with $P_i$ the projection over the support of $\rho_i$,
then $K$ can be uniquely identified by a set of projections $\{P_i\}$ 
(up to a permutation over inputs). In this case,
\[
  K=\operatorname{span}\bigl\{ \ket{i}\!\bra{\psi} : \ket{\psi} \in \operatorname{supp}\,\rho_i \bigr\},
  \quad
  P=\sum_i \proj{i} \otimes P_i.
\]
Any such Kraus operator space will be 
called ``non-commutative bipartite cq-graph'' or simply ``cq-graph''.

The case of assisted communication seems very complicated, and most interesting. 
We show that the zero-error classical capacity of cq-graphs is given by the solution 
of the following  SDP:
\begin{equation}
  \label{eq:Aram-cq}
  \Aram(K) := \max \sum_i s_i \ {\rm s.t. }~0\leq s_i,\ \sum_i s_i P_i\leq \1.
\end{equation}
This number was suggested by Harrow as a natural generalization of the 
Shannon's classical fractional packing number \cite{Harrow2010}, 
and we will refer to it as \textit{semidefinite (fractional) packing number} 
associated with a set of projections $\{P_i\}$.

Our result can be summarized as 
\begin{theorem}
  \label{thm:cq-capacity}
  The zero-error classical capacity of a cq-channel $\cN: i\rightarrow \rho_i$ 
  assisted by quantum no-signalling correlations is given by the logarithm 
  of the semidefinite packing number $\Aram(K)$, i.e.,
  \[
    C_{0,{\rm NS}}(K)=\log \Aram(K).
  \]
  To be precise,
  \[
    \frac{1}{\text{\rm poly}(n)} \Aram(K)^n \leq \U\left(K^{\ox n}\right) \leq \Aram(K)^n.
  \]
\end{theorem}
The proof of this theorem is given in Section~\ref{subsec:asymptotic-capacity-cq}.

The asymptotic zero-error classical simulation cost for cq-graphs $K$ is 
relatively easy and straightforward. Indeed, we show 
that the one-shot zero-error classical simulation cost for cq-channels 
is multiplicative, \emph{i.e.}
\begin{equation}
  \label{eq:sigma-mult}
  \S(K_1\ox K_2)=\S(K_1)\,\S(K_2),
\end{equation}
for cq-graphs $K_1$ and $K_2$. The equality is proved
by simply combining the sub-multiplicativity of the primal SDP and the 
super-multiplicativity of the dual problem, and then applying strong 
duality of SDPs. It readily follows that the asymptotic simulation cost for any 
Kraus operator $K$ corresponding to a cq-channel is given by the one-shot 
simulation cost, namely
\begin{equation}
  \label{eq:G0-equals-sigma}
  S_{0,\NS}(K) = \log \S(K).
\end{equation}
It is worth noting that the above equality is valid for many other cases. 
In particular, it holds when $K$ corresponds to a quantum channel that 
is an extreme point in the set of all quantum channels. It remains unknown 
whether this is true for more general $K$.

As an unexpected byproduct of our general analysis we obtain the following.
To understand it, note that any cq-channel $\cN:i \longmapsto \rho_i$ 
naturally induces a confusability graph $G$ on the vertices $i$, by
letting $\{i,j\}\in E$ if and only if $\rho_i\not \perp \rho_j$, i.e.
inputs $i$ and $j$ are confusable. 

\begin{theorem}
  \label{thm:theta}
  For any classical graph $G$, the Lov\'asz number $\vartheta(G)$ \cite{Lovasz1979} 
  is the minimum zero-error classical capacity assisted 
  by quantum no-signalling correlations of any cq-channel that induces $G$, \emph{i.e.}
  \[
    \log\vartheta(G)=\min\bigl\{ C_{0,{\rm NS}}(K): K^\dag K<S_G \bigr\},
  \]
  where the minimization is over cq-graphs $K$ and $S_G$ is the non-commutative graph corresponding to $G$ (see Eq. (\ref{sg-graph}) in Section V).

  In particular, equality holds for any cq-channel $i\rightarrow \proj{\psi_i}$ 
  such that $\{\ket{\psi_i}\}$ is an optimal orthogonal representation for $G$
  in the sense of Lov\'{a}sz' original definition \cite{Lovasz1979}. 
\end{theorem}

The above result gives the first clear information theoretic interpretion
of the Lov\'asz $\vartheta$ function of a graph $G$, as the zero-error classical 
capacity of $G$ assisted by quantum no-signalling correlations, in the
sense of taking the worst cq-channel with confusability graph $G$. 
Its proof will be given in Section \ref{sec:theta}.

In Section \ref{sec:feasibility}, we provide results which completely characterize the 
feasibility (i.e., positivity) of zero-error communication with a 
non-commutative bipartite graph, or a quantum channel, assisted by no-signalling
correlations.  
Finally, in Section \ref{sec:conclusion},
we conclude and propose several open problems for further study.

\section{Structure of quantum no-signalling correlations}
\label{sec:structure}

\subsection{Do quantum no-signalling correlations naturally occur in communication?}
\label{subsec:naturally}
We will provide an intuitive explanation for how the quantum no-signalling 
correlations naturally arise from the information-theoretic viewpoint.
Let $\cM$ and $\cN$ be two quantum channels both from A to B, and assume 
that A and B can access any quantum resources that cannot be used for 
communicating between them directly. An interesting question is to 
ask when $\cN$ can exactly simulate $\cM$? This class of resources 
clearly includes shared entanglement, and actually has some other 
more important members that are the central interest of this paper. 
We will derive some general constraints that all these resources should 
satisfy. Let us start with the one-shot case, that is, A and B 
can establish quantum channel $\cM$ by using all ``allowable resources'' 
and one use of channel $\cN$. We can abstractly represent the whole 
procedure as in Fig.~\ref{fig:simulation}; note that it falls within
the formalism of ``quantum combs''~\cite{CD-AP}, but as it is a rather
special case we can understand it without explicitly invoking that
theory.

\begin{figure}[ht]
  \includegraphics[width=11cm]{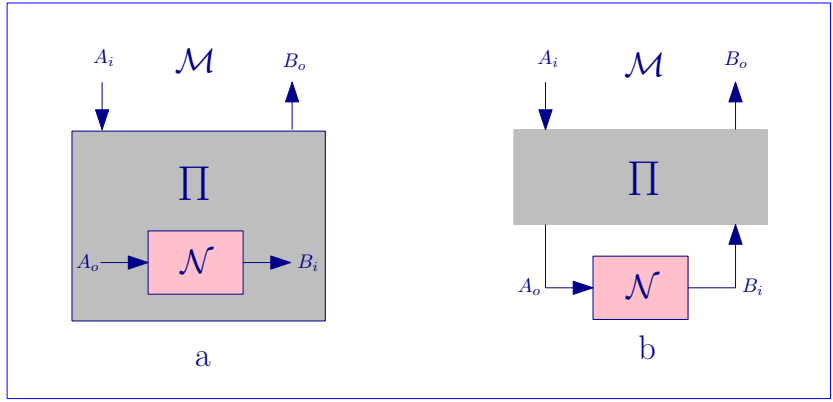}
  \caption{A general simulation network: a). We have abstractly represented the general simulation procedure for implementing a channel $\cM$ using another channel $\cN$ just once, and the correlations between A and B; b). This is just an equivalent way to redraw a), and we have highlighted all correlations between A an B, and their pre- and/or post- processing as $\Pi$.}
  \label{fig:simulation}
\end{figure}
Here is how this simulation works. First, A performs some pre-processing on 
the input quantum system $A_i$ together with all possible resources at her hand, 
and outputs quantum system $A_o$ as the input of channel $\cN$. The channel $\cN$ 
then outputs a quantum system $B_i$. B will do some post-processing on this 
system together with all possible resources he has, and finally generates an 
output $B_o$. If we remove $\cN$, we are left with a network with two inputs 
$A_i$ and $B_i$, and two outputs $A_o$ and $B_o$. Clearly this represents all 
possible pre- and/or post-processing that A and B have done and all resources 
that are available to A and B. In the framework of quantum mechanics, this 
network can be formulated as a quantum channel $\Pi$ with two inputs and two 
outputs. Thus we can redraw the simulation procedure as Fig.~\ref{fig:simulation}.b.  
However, as $\cN$ is the only communicating device from A to B, we must have 
that $\Pi$ cannot be used to communicate from A to B. Furthermore, the output 
$A_o$ represents the input of A to the channel $\cN$, and thus can be uniquely 
determined by $A$, but not $B_i$, which is the output of $\cN$. We will see this 
constraint is equivalent to B cannot communicate to A using $\Pi$.  
These constraints have led us to a fruitful class of resources that A and B 
could use in communication. 

As before, we don't attempt to solve the most general channel simulation problem. 
Instead, we will focus on two simpler but most interesting cases: 
i). $\cM$ is a noiseless classical channel and $\cN$ is the given noisy channel. 
The optimal solution to this problem will lead us to the notion of zero-error 
classical capacity of $\cN$;  
ii). $\cN$ is a noiseless classical channel and $\cM$ is the given noisy channel. 
The optimal solution will lead us to the notion of zero-error classical simulation 
cost of $\cM$.  In the communication problem, we want to maximize the number 
of messages we can send exactly by the given channel; while in the simulation 
problem, we want to minimize the amount of the  noiseless classical communication 
to simulate the given channel. In the rest of this section,  we will study the 
mathematical structures of quantum no-signalling correlations in detail.

\subsection{Mathematical definition of quantum no-signalling correlations}
\label{sec:definitions}
As discussed before, quantum no-signalling correlations are linear maps
\[
  \Pi: \cL(A_i) \ox \cL(B_i)\rightarrow \cL(A_o) \ox \cL(B_o)
\]
with additional constraints.  First, $\Pi$ is required to be completely positive (CP) and trace-preserving (TP). This makes $\Pi$ a physically realizable quantum operation. Furthermore, $\Pi$ is A to B no-signalling (A$\not\rightarrow$B). That is, A cannot send classical information to B by using $\Pi$. More precisely, for any density operators $\rho^{(0)}_{A_i}, \rho^{(1)}_{A_i}\in \cL(A_i)$ and $\sigma_{B_i}\in \cL(B_i)$, we have
$$\tr_{A_o} \Pi(\rho^{(0)}_{A_i}\ox \sigma_{B_i})=\tr_{A_o} \Pi(\rho^{(1)}_{A_i}\ox \sigma_{B_i}).$$
Or equivalently, 
\[
  \tr_{A_o} \Pi(X_{A_i}\ox Y_{B_i})=0\ \forall X,\ Y \text{ s.t. }\tr X=0.
\]
Likewise, $\Pi$ is required to be B to A no-signalling (B$\not\rightarrow$A). That is, B cannot send classical information to A by using $\Pi$. This constraint can be formulated as the following
\[
  \tr_{B_o} \Pi(X_{A_i}\ox Y_{B_i})=0\ \forall X,\ Y \text{ s.t. }\tr Y=0.
\]
Let the Choi-Jamio\l{}kowski matrix of $\Pi$ be 
$$\O_{A_i' A_o B_i' B_o} = (\id_{A_i'} \ox \id_{B_i'} \ox \Pi)(\Phi_{A_i A_i'} \ox \Phi_{B_i B_i'}),$$
where $\id_{A_i'}$ is the identity operator over $\cL(A_i')$, $\Phi_{A_i A_i'}=\ketbra{\Phi_{A_iA_i'}}{\Phi_{A_iA_i'}}$,  and $\ket{\Phi_{A_iA_i'}}=\sum_{k}\ket{k_{A_i}}\ket{k_{A_i'}}$ the un-normalized maximally entangled state. We now show that all above constraints on $\Pi$ can be easily reformulated into the semidefinite programming constraints in terms of the Choi-Jamio\l{}kowski matrix $\Omega$. For convenience, we often use unprimed letters such as $A_i$ and $B_i$ to denote the quantum systems inputting to quantum channels, and the primed letters $A_i'$ and $B_i'$  for the reference systems which are isomorphic to $A_i$ and $B_i$, respectively.  The constraints on $\Pi$ can be equivalently formulated in terms of $\Omega$ as follows: 
\begin{align*}
\O \ge 0, &\quad (CP)\\
\tr_{A_o B_o} \O = \1_{A_i' B_i'}, &\quad (TP)\\
\tr_{A_o A_i'} \O X_{A_i'}^{T} = 0\ \forall \tr X = 0, &\quad ({\rm A} \not\rightarrow {\rm B})\\
\tr_{B_o B_i'} \O Y_{B_i'}^{T} = 0\ \forall \tr Y = 0, &\quad ({\rm B} \not\rightarrow {\rm A}) 
\end{align*}
where $X$ and $Y$ are arbitrary Hermitian operators, so the transpose is not 
really necessary. 
The first two constraints guarantee that $\O$ corresponds to a CPTP map $\Pi$, while the latter two make sure that $\O$ cannot be used for communicating from A to B and B to A, respectively. Both constraints need only be verified on a Hermitian matrix basis of $A_i'$, $B_i'$, respectively.

The key to deriving the above constraints is the following useful fact:
$$X_{A}=\tr_{A'}\Phi_{AA'}(\1_{A}\ox {X^T_{A'}})=\tr_{A'} \Phi_{AA'}{X^T_{A'}},$$
where $A'$ is an isomorphic copy of $A$.

It is worth noting that the class of quantum no-signalling correlations is closed under convex combinations. That is, if $\Pi_0$ and $\Pi_1$ are quantum no-signalling correlations and $0\leq p\leq 1$, then $p\Pi_0+(1-p)\Pi_1$ are also no-signalling correlations.  Furthermore,  this class is also stable under the pre- or post-processing by A or B. That is, if $\Pi$ is a no-signalling correlation from $\cL(A_i'\ox B_i')$ to $\cL(A_o\ox B_o)$. Then $\Pi'=(\cA_1\ox \cB_1)\Pi (\cA_0\ox \cB_0)$ is also no-signalling, where $\cA_0,\cA_1,\cB_0,\cB_1$ are CPTP maps on suitable Hilbert spaces.

It is instructive to compare the quantum no-signalling correlations and the classical  
no-signalling correlations. Recall that any classical no-signalling correlations 
can be described as a classical channel $Q=(X\times Y, Q(ab|xy), A\times B)$ 
with two classical inputs $x\in X$ and $y\in Y$, and two classical outputs 
$a\in A$ and $b\in B$, where 
\begin{align}
  Q(ab|xy)          &\geq 0, \ \forall x\in X, y\in Y, a\in A,\ b\in B,          \\
  \sum_{ab}Q(ab|xy) &= 1, \ \forall x\in X,\ y\in Y         ,          \\
  \sum_{a}Q(ab|xy)  &= \sum_{a}Q(ab|x'y), \ \forall x,x'\in X,\ y\in Y,\\
  \sum_{b}Q(ab|xy)  &= \sum_{b}Q(ab|xy'), \ \forall x\in X,\ y,y'\in Y.
\end{align}
Evidently, $Q$ can also be represented as a quantum channel in the following way
\[
  Q(\rho)=\sum_{a,b,x,y} Q(ab|xy) \ketbra{ab}{xy} \rho \ketbra{xy}{ab}.
\]
One can easily verify the above constraints are exactly  the same as treating 
$Q$ a quantum no-signalling correlation. From this viewpoint, quantum 
no-signalling correlations are natural generalizations of their classical correspondings. 

Finally, we would like to mention another interesting fact. That is, any two-input and two-output quantum channel $\Pi$ can be reduced to a two-input and two-output classical channel $Q$ that has the same signalling property by simply doing pre- or post-processing, and all inputs and outputs of $Q$ are binary, i.e., if $\Pi$ is A to B (and/or B to A) signalling then $Q$ is also A to B (resp. B to A) signalling. Due to its significance, we formulate it as
\begin{proposition}
  \label{binary-signalling}
  For any CPTP map $\Pi$ from $\cL(A_i\ox B_i)$ to $\cL(A_o\ox B_o)$ such that $\Pi$ is 
  B to A (and/or A to B) signalling, one can obtain a classical channel 
  $Q=(X\times Y, Q(ab|xy), A\times B)$ with all $|A|=|B|=|X|=|Y|=2$, by doing 
  suitable local pre- and post-processing on $\Pi$, such that $Q$ is also B 
  to A signalling (resp. A to B). 
\end{proposition}

\begin{proof}
Assume $\Pi$ is both way signalling (the case that it is one-way signalling
is similar, and in fact simpler). Then we can find a pair of states 
$\rho^{(0)}_{A_i}$ and $\sigma^{(0)}_{B_i}$, such that the following two 
maps $\cA: \cL(A_i)\rightarrow \cL(B_o)$ and 
$\cB: \cL(B_i)\rightarrow \cL(A_o)$ are \textit{non-constant} CPTP maps:
\begin{align*}
  \cA(\rho_{A_i})   &= \tr_{A_o}\Pi(\rho_{A_i}\ox \sigma^{(0)}_{B_i}),   \\
  \cB(\sigma_{B_i}) &= \tr_{B_o}\Pi(\rho^{(0)}_{A_i}\ox \sigma_{B_i}).
\end{align*}
I.e., we can find another two states $\rho^{(1)}_{A_i}$ and $\sigma^{(1)}_{B_i}$ such that
$$\cA(\rho^{(0)}_{A_i})\neq \cA(\rho^{(1)}_{A_i}), ~~\cB(\sigma^{(0)}_{B_i})\neq \cB(\sigma^{(1)}_{B_i}).$$

By Helstrom's theorem, for any two different quantum states 
$\{\rho_0, \rho_1\}$, we can find a projective measurement 
$\{P_0, P_1\}$  to distinguish the given states with equal prior 
probably with an average success probability
$$p = \frac{1}{2}\tr \rho_0 P_0 + \frac{1}{2}\tr \rho_1 P_1 
    = \frac{1}{2}+\frac14 \|\rho_0-\rho_1\|_1>1/2.$$
As a direct consequence, the CPTP map 
$$\cM(\rho)=(\tr \rho P_0) \ketbra{0}{0}+(\tr \rho P_1) \ketbra{1}{1}$$
will produce two different classical binary probability distributions 
$\cM(\rho_0)$ and $\cM(\rho_1)$. We will need this fact below.

Let $\bigl(P_{A_o}^{(0)}, P_{A_o}^{(1)}\bigr)$ and 
$\bigl(Q_{B_o}^{(0)}, Q_{B_o}^{(1)}\bigr)$ 
be the measurements to optimally distinguish 
$\bigl\{\cB(\sigma^{(0)}_{B_i}), \cB(\sigma^{(1)}_{B_i})\bigr\}$ and
$\bigl\{\cA(\rho^{(0)}_{A_i}), \cA(\rho^{(1)}_{A_i})\bigr\}$, respectively.
Then we can define four CPTP maps as follows:
\begin{align*}
  \cA_0(\rho)    =(\tr\ketbra{0}{0}\rho)\rho_{A_i}^{(0)}+(\tr\ketbra{1}{1}\rho)\rho_{A_i}^{(1)},
                  &\quad
  \cA_1(\rho)    =(\tr P_{A_o}^{(0)}\rho)\ketbra{0}{0}+(\tr P_{A_o}^{(1)}\rho)\ketbra{1}{1},\\
  \cB_0(\sigma)  =(\tr\ketbra{0}{0}\sigma)\sigma_{B_i}^{(0)}+(\tr\ketbra{1}{1}\sigma)\sigma_{B_i}^{(1)},
                  &\quad
  \cB_1(\sigma)  =(\tr Q_{B_o}^{(0)}\sigma)\ketbra{0}{0}+(\tr Q_{B_o}^{(1)}\sigma)\ketbra{1}{1}.
\end{align*}
Using these as pre- and post-processing on $\Pi$, we obtain the desired channel
$$Q=(\cA_1\ox \cB_1)\circ \Pi\circ (\cA_0\ox \cB_0).$$
In $Q$, if B inputs $0$ or $1$, and A inputs $0$, then by the above construction, 
A must output two binary probability distributions that are different, 
hence can be used for signalling from B to A. 
Similarly, if A inputs $0$ or $1$, and B inputs $0$, then B must output another 
two different binary probability distributions that can be used for signalling 
from A to B.
\end{proof}

\subsection{Structure theorems for quantum no-signalling correlations}
\label{subsec:structure}
We will establish several structure  theorems regarding quantum no-signalling correlations. 
Note that $\Pi$ is a two-input and two-output quantum channel. 
So there are two natural ways to think of $\Pi$ according to the relation 
between the inputs and outputs. 
The first way is to partition $\Pi$ as $A_iA_o:B_iB_o$, so the output of $A_i$ 
and $B_i$ would be $A_o$ and $B_o$, respectively; this is perhaps the most standard way. 
The second way is to focus on the communication between $A$ and $B$
and partition $\Pi$ as $A_iB_o:B_iA_o$. In this case the output of $A_i$ and
 $B_i$ will be regarded as $B_o$ and $A_o$, respectively. This kind of partition 
 is quite useful when we are interested in the communication between $A$ and $B$. 

Let us start with the bipartition $A_iA_o:B_iB_o$ of $\Pi$. 
In this case, we can have a full characterization of no-signalling maps. 
(Bear in mind that we have assumed $\Pi$ to be a CPTP map in the following discussion).

\begin{proposition}
\label{ns-structure1}
Let $\Pi: \cL(A_i) \ox \cL(B_i)\rightarrow \cL(A_o) \ox \cL(B_o)$ be a bipartite CPTP map with a decomposition 
$\Pi=\sum_{k}\lambda_k \cA_k\ox \cB_k$ according to bipartition $A_iA_o:B_iB_o$,  where $\cA_k:\cL(A_i)\rightarrow \cL(A_o)$, $\cB_k:\cL(B_i)\rightarrow \cL(B_o)$,  and $\lambda_k$ are complex numbers. Then we have the following:

i) $\Pi$ is B to A no-signalling iff  $\cB_k$  can be chosen as CPTP for any $k$.
 In this case $\sum_k\ll_k \cA_k$ will be also CPTP.

ii) $\Pi$ is A to B no-signalling iff $\cA_k$ can be chosen as CPTP for any $k$.
 In this case $\sum_k\ll_k \cB_k$ will be also CPTP.

iii) $\Pi$ is no-signalling between A and B iff both $\cA_k$ and $\cB_k$ can be 
 chosen as CPTP maps, $\sum_k \ll_k=1$, and $\lambda_k$ are real for all $k$. 

\end{proposition}
The most interesting part is item iii).  Intuitively, any no-signalling correlation $\Pi$ between A and B can be written as a real affine combination of product CPTP maps $\cA_k \otimes \cB_k$. The proofs are relatively straightforward and we simply leave them to interested readers as exercises.

Now we turn to the bipartition $A_iB_o:B_iA_o$ of $\Pi$. We don't have a full characterization of no-signalling correlations in this setting. Nevertheless, a simple but extremely useful class of no-signalling correlations can be constructed using the following facts.
\begin{proposition}\label{ns-structure2}
Let $\Pi: \cL(A_i) \ox \cL(B_i)\rightarrow \cL(A_o) \ox \cL(B_o)$ be a bipartite CPTP map with a  decomposition $\Pi=\sum_{k}\mu_k \cE_k\ox \cF_k$ according to bipartition $A_iB_o:B_iA_o$, where $\cE_k:\cL(A_i)\rightarrow \cL(B_o)$, $\cF_k:\cL(B_i)\rightarrow \cL(A_o)$,  and $\mu_k$ are complex numbers. Then we have the following:

i) If $\sum_k \mu_k\cE_k$  is a constant map and $\cF_k$ is TP for any $k$, then $\Pi$ is A to B no-signalling.

ii) If $\sum_k \mu_k\cF_k$ is a constant map and $\cE_k$ is TP for any $k$, then $\Pi$ is B to A no-signalling.

iii) If all $\cE_k$ and $\cF_k$ are TP and both $\sum_k \mu_k\cE_k$ and $\sum_k \mu_k\cF_k$ are constant maps, then $\Pi$ is no-signalling between A and B.
\end{proposition}

For example, we can choose $\Pi=\sum_k \mu_k \cE_k\ox \cF_k$, where $\cE_k$ and $\cF_k$ are CPTP maps from $\cL(A_i)$ to $\cL(B_o)$ and $\cL(B_i)$ to $\cL(A_o)$, respectively, and $\{\mu_k\}$ is a probability distribution. If we further have that $\sum_k \mu_k\cE_k$ and $\sum_k \mu_k\cF_k$ are constant maps, then we obtain a no-signalling correlation $\Pi$. Clearly, for any such $\Pi$, neither of $A_i$ and $B_i$ can send classical information to $A_o$ or $B_o$. To emphasize this special feature, this class of no-signalling correlations are said to be \textit{totally no-signalling}. We shall see later even this class of correlations could be very useful in assisting zero-error communication and simulation. In particular, we need the following technical result:
\begin{lemma}\label{rank2-ns}
Let $\cE_0$ and $\cE_1$ be two CPTP maps from $\cL(A_i)$ to $\cL(B_o)$, and let $\cF_0,\cF_1$ be two CP maps from $\cL(B_i)$ to $\cL(A_o)$. Furthermore, assume there is unique $0\leq p\leq 1$ such that $p\cE_0+(1-p)\cE_1$ is a constant channel. Then $\Pi=p\cE_0\ox \cF_0+(1-p)\cE_1\ox \cF_1$ is a no-signalling correlation if and only if both $\cF_0$ and $\cF_1$ are CPTP maps, and $p\cF_0+(1-p)\cF_1$ is a constant map. 
\end{lemma}

\begin{proof}
This can be proven by directly applying the definition of quantum no-signalling 
correlations. Two key points are: First, the uniqueness of $p$ such that $p\cE_0+(1-p)\cE_1$ 
is constant; secondly, if $\cE(X)=0$ for any traceless Hermitian operator $X$, 
$\cE$ is a constant map.
\end{proof}

A third result about the structure of no-signalling correlations is the following \cite{ESW2001}
(see also \cite{PHHH2006} for an alternate proof, and \cite{Chiribella2012} for more discussions):

\begin{proposition}
\label{ns-structure3}
$\Pi$ is B$\not\rightarrow$A iff $\Pi=\cG\circ \cF$ for two CPTP maps 
$\cF: \cL(A_i)\rightarrow \cL(A_o\ox R)$ and $\cG: \cL(B_i\ox R)\rightarrow \cL(B_o)$, 
where $R$ is an internal memory.
\qed
\end{proposition}
This result is quite intuitive. It says that in the case that B cannot 
communicate to A, we can actually have the output of A before we input to B. 

\subsection{Composing no-signalling correlations with quantum channels}
\label{subsec:composition}
Suppose now we are given a no-signalling map $\Pi:\cL(A_i\ox B_i)\rightarrow \cL(A_o\ox B_o)$
and a CPTP map $\cN: \cL(A)\rightarrow \cL(B)$. 
When $A=A_o$ and $B=B_i$, we can construct a new and unique map 
$\cM: \cL(A_i)\rightarrow \cL(B_o)$ by feeding the output $A_o$ of 
$\Pi$ into $\cN$, and use the output $B$ of $\cN$ as the input $B_i$ to 
$\Pi$ (refer to Fig. 2). Next we will show how to derive more explicit forms of this map.

\begin{figure}[ht]
  \includegraphics[width=11cm]{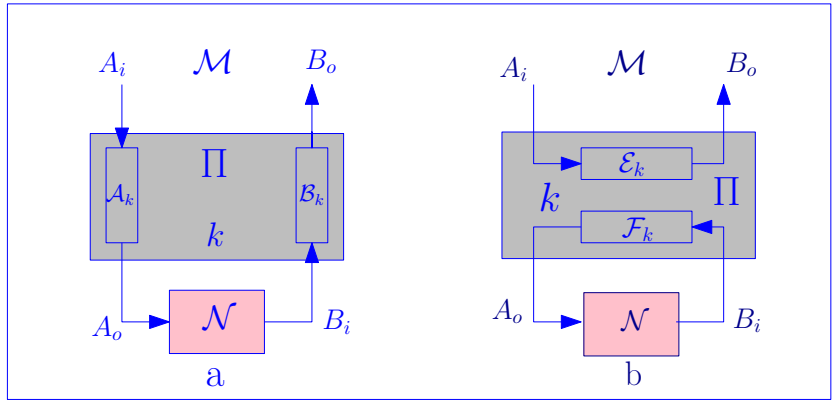}
  \caption{Two different partitions of the input-output registers of $\Pi$}
  \label{fig:partitions}
\end{figure}

Let us start by considering the partition according to $A_iA_o:B_iB_o$ to look
at $\Pi$. If $\Pi$ is of product form $\cA\ox \cB$ where 
$\cA: \cL(A_i)\rightarrow \cL(A_o)$ and $\cB: \cL(B_i)\rightarrow \cL(B_o)$, 
then clearly the composition must be given by 
$\cM(X)=\cB\circ \cN\circ \cA(X)$. 
We can simply extend this construction to the superpositions of these product maps
by linearity. I.e., for a general $\Pi=\sum_k\cA_k\ox \cB_k$, we define the 
new map by composing $\cN$ and $\Pi$ as the following
\begin{equation}\label{connection-1}
  \cM^{A_i\rightarrow B_o}
    =\sum_k \cB_k^{B_i\rightarrow B_o}\circ \cN^{A_o\rightarrow B_i}\circ \cA_k^{A_i\rightarrow A_o},
\end{equation}
and this is clearly well-defined.

We can also formulate $\cM$ in another useful way by considering the bipartition $A_iB_o:B_iA_o$. 
Suppose $\Pi$ is given in the form of $\sum_{k} \cE_k\ox \cF_k$  where 
$\cE_k:\cL(A_i)\rightarrow \cL(B_o)$, $\cF_k:\cL(B_i)\rightarrow \cL(A_o)$. 

Then to compose $\Pi$ and $\cN$, the only thing we need to do is to compose $\cN$ and $\cF_k$ directly and take trace. That is, we have
\begin{equation}\label{connection-2}
  \cM=\sum_k \cE_k^{A_i\rightarrow B_o} \tr(\cF_k\circ \cN),
\end{equation}
where the trace of super-operator $\cF_k\circ \cN$ on $\cL(A_o)$ is given by
\[
  \tr \cF_k\circ \cN = \sum_j \tr (C_j^\dagger (\cF_k\circ\cN)(C_j))
\]
for any orthonormal basis $\{C_j\}$ on $\cL(A_o)$ in the sense of $\tr (C_j^\dagger C_l)=\delta_{jl}$.  
This is just the natural generalization of the trace function of the square matrices
to super-operators.  
It is easy to verify that $\tr (\cF\circ \cN) = \tr (\cN\circ \cF)$ holds for any 
$\cF$ and $\cN$, and the trace is independent of the choice of the orthonormal basis $\{C_j\}$.

For instance, for any bipartite classical channel $Q=(X\times Y, Q(ab|xy), A\times B)$, and another classical channel $N=(A, N(y|a), Y)$, we can easily show that the map $M=(X, M(b|x), B)$ constructed by composing $Q$ and $N$ is given by
$$M(b|x)=\sum_{y,a}Q(ab|xy)N(y|a),$$
which coincides with the discussions in \cite{CLMW2011}.

Note that the above two constructions of the map $\cM$ from $\Pi$ and $\cN$ are quite general, and purely mathematical. Indeed, we can form a map $\cM$ from any $\cN$ and $\Pi$. An interesting and important question is to ask whether $\cM$ is a CPTP map when both $\cN$ and $\Pi$ are. 

It turns out that whenever $\cN$ and $\Pi$ are CP maps, $\cM$ should also be CP. This fact is actually a simple corollary of the teleportation protocol.  Indeed, we can easily see from the Fig. 3 that 
\begin{equation}
  \label{connection-3}
  \cM(X_{A_i})=\tr_{BB_i'}\{ [((\id_{B_oB_i'}\ox \cN)\circ(\id_{B_i'}\ox \Pi))
                            (X_{A_i}\ox \Phi_{B_iB_i'})](\1_{B_o}\ox\Phi_{BB_i'})\},
\end{equation}
where $A=A_o$, and $B$, $B_i$, and $B_i'$ are all isomorphic. The way to get the above equality is to first apply the teleportation protocol to the special case that $\Pi=\cA\ox \cB$, and then extend the result to the general case by linearity. The above expression not only provides a proof for the complete positivity of $\cM$, but also the uniqueness of such $\cM$.

\begin{figure}[ht]
  \includegraphics[width=9cm]{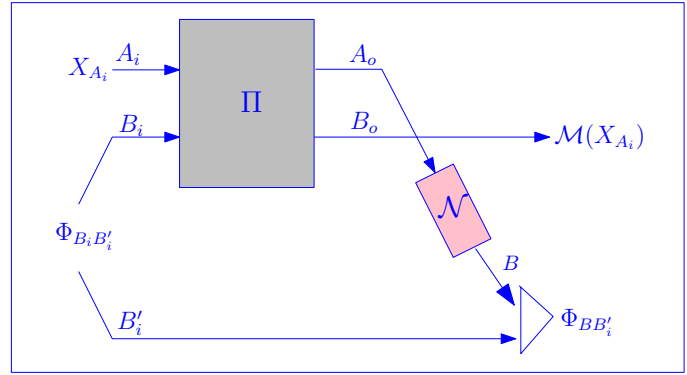}
  \caption{Composing $\cN$ and $\Pi$ from the viewpoint of teleportation.}
  \label{fig:teleport}
\end{figure}

So we have obtained that for any two maps $\Pi: \cL(A_i\otimes B_i)\rightarrow \cL(A_o\ox B_o)$ and $\cN: \cL(A_o)\rightarrow \cL(B_i)$, we can uniquely construct a map $\cM:\cL(A_i)\rightarrow \cL(B_o)$ by composing $\Pi$ and $\cN$, and trace out $A_o$ and $B_i$. An explicit construction can be equivalently given by any of Eqs. (\ref{connection-1}), (\ref{connection-2}), and (\ref{connection-3}). Furthermore, if $\Pi$ and $\cN$ are CP, so is $\cM$.

To guarantee that $\cM$ is TP, we have to put further constraints on $\Pi$, but not on $\cN$ except that it is an arbitrary CPTP map. We will show that such constraints are precisely the same as the constraints that B to A no-signalling (B$\not\rightarrow$ A). 

\begin{proposition}
Let $\Pi: \cL(A_i\otimes B_i)\rightarrow \cL(A_o\ox B_o)$ and $\cN: \cL(A_o)\rightarrow \cL(B_i)$ be two CPTP maps, and let $\cM:\cL(A_i)\rightarrow \cL(B_o)$ be the unique CP constructed by composing $\Pi$ and $\cN$. If $\Pi$ is CPTP and {\rm B} to {\rm A} no-signalling, then $\cM$ is also CPTP for any CPTP map $\cN$. Conversely, if $\cM$ is CPTP for any $\cN$, then $\Pi$ has to be {\rm B} to {\rm A} no-signalling. 
\end{proposition}

\begin{proof}
The sufficiency is just a direct corollary of item i) of Proposition \ref{ns-structure1} or Proposition \ref{ns-structure3}. The later actually provides a physical realization of this map $\cM$. This constraint is logically reasonable as B to A no-signalling simply means that the input of the channel $\cN$ cannot depend on its output (otherwise we will have a closed loop, against the causality principle). 

Now we will prove that the constraint of no-signalling from B to A is also necessary 
for $\cM$ to be CPTP for any CPTP $\cN$.  Our strategy is to first prove this fact for two-input and two-output binary classical channel $Q$. Then  we can apply Proposition \ref{binary-signalling} for the general case.

We can actually show that if $Q$ is signalling from B to A, and the composition of $Q$ and one-bit noiseless classical channel is not a legal classical channel anymore (the output for some input is not a legal probability distribution). 

By contradiction, assume that the composition of $\Pi$ and any CPTP map $\cN$ is always 
CPTP, but that $\Pi$ is B to A signalling. Applying Proposition \ref{binary-signalling}, we can find a two-input and two-output binary classical channel $Q=(X\times Y, Q(ab|xy), A\times B)$ such that $Q$ is constructed from $\Pi$ by simply performing pre- and/or post-processing such that  i) $Q=(\cA_1\ox \cB_1)\circ \Pi\circ (\cA_0\ox \cB_0)$;  ii) $Q$ is also B to A signalling; iii) the composition of $Q$ and any classical channel $A$ to $Y$ is also a legal classical channel. We will show such $Q$ does not exist, hence complete the proof. Actually, ii) implies there are $x\in A, a\in A, y,y'\in Y$, such that 
$$\sum_{b=0}^1 Q(ab|xy)\neq \sum_{b=0}^1 Q(ab|xy').$$
w.o.l.g, we can assume $x=0$, $a=0$, $y=0$, and $y'=1$. Then the above inequality can be rewritten explicitly as 
$$Q(00|00)+Q(01|00)\neq Q(00|01)+Q(01|01).$$
On the other hand, the composition of $Q$ and a one-bit noiseless classical channel $N(y|a)=\delta_{ya}$ is given by
$$M(b|x)=\sum_{ya} Q(ab|xy)N(y|a)=\sum_{a=0}^1 Q(ab|xa)=Q(0b|x0)+Q(1b|x1).$$
In particular, taking $x=0$ and applying the above inequality, we have 
\[\begin{split}
  \sum_{b=0}^1 M(b|0) &=   Q(00|00)+Q(01|00)+Q(10|01)+Q(11|01) \\
                      &\neq Q(00|01)+Q(01|01)+Q(10|01)+Q(11|01)=1.
\end{split}\]
This contradicts the fact that $M(b|x)$ is a classical channel.
\end{proof}

Thus, to guarantee that the composition $\cM$ of a two-input and two-output channel $\Pi$ and any quantum channel from A to B is always a CPTP map, $\Pi$ is only required to be B to A no-signalling. The constraint of no-signalling from A to B is another natural requirement as we have assumed that A cannot communicate to B directly, and the given channel $\cN$ is the only directed resource from A to B we can use. These considerations naturally lead us to quantum no-signalling correlations. 

\section{Semidefinite programmes for zero-error communication \protect\\ and simulation assisted by quantum no-signalling correlations}
\label{sec:comm+sim-SDPs}

\subsection{Zero-error assisted communication capacity}
\label{subsec:1-shot-capacity}
For any integer $M$, we denote by $M_a$ and $M_b$ classical registers with size $M$.  
We use $\cI_M: M_a\rightarrow M_b$ to denote the noiseless classical channel 
that can send $M$ messages from $A$ to $B$, \emph{i.e.}
$\cI_M(\ketbra{m}{m'}_{M_a})=\delta_{mm'} \ketbra{m}{m}_{M_b}$, or equivalenyly,
\[
  \cI_M(\rho) = \sum_{m=1}^M \bigl( \tr \rho \ketbra{m}{m}_{M_a} \bigr) \ketbra{m}{m}_{M_b}.
\]

If a channel $\cN$ can simulate $\cI_M$, we say that $M$ messages  
can be perfectly transmitted by one use of $\cN$. 
The problem we are interested in is to determine the largest possible $M$, 
which will be called one-shot zero-error classical transmission capability 
of $\cN$ assisted with no-signalling correlations. 
We restate here the result that we will prove in this subsection.

\medskip\noindent
{\bf Theorem~\ref{capacity-1-shot}}\ 
\textit{The one-shot zero-error classical capability (quantified as the largest
  number of messages) of 
  $\cN$ assisted by quantum no-signalling correlations depends only on
  the non-commutative graph $K$, and is given by the 
  integer part of the following SDP:
  \[\begin{split}
    \U(\cN) & =\U(K) \\
            & =\max \tr S_A\ \text{ s.t. }\ 0\leq E_{AB}\leq S_A\ox \1_B, 
                                         \tr_A E_{AB}=\1_B, \tr P_{AB}(S_A\ox \1_B-E_{AB})=0,
  \end{split}\]
    where $P_{AB}$ denotes the projection onto the the subspace $(\1\ox K)\ket{\Phi}$. \\
  Hence we are motivated to call $\U(K)$ the 
  \emph{no-signalling assisted independence number of $K$.}}

\medskip
\begin{proof}
To make the indices of the Choi-Jamio\l{}kowski matrices more readable, 
we reserve the unprimed letters such as $A_i$ and $B_i$ for the reference systems 
(thus as the indices of Choi-Jamio\l{}kowski matrices), and use the primed versions 
$A_i'$ and $B_i'$ as the inputs to the quantum channels.

We will use a general no-signalling correlation $\Pi$  with Choi-Jamio\l{}kowski matrix $\O$  and the channel $\cN$, to exactly simulate a noiseless classical channel $\cI_M$. Our goal is to determine the maximum integer $M$ when such simulation is possible. In this case  both $A_i' = M_a'$ and $B_o = M_b$ are classical, and $A_o = A'$, $B_i' = B$. We will show that $\O$ can be chosen as the following form:
\begin{equation}\label{omega-cns0}
\O_{M_aM_bAB} = \frac{1}{M}D_{M_aM_b} \ox E_{AB} + \frac{1}{M}(\1-D)_{M_aM_b} \ox F_{AB},
\end{equation}
where $D=\sum_m \ketbra{mm}{mm}$ is the Choi-Jamio\l{}kowski matrix of the noiseless classical channel $\cI_M$, and $E$ and $F$ are positive semidefinite operators. We will show that the no-signalling conditions are translated into
\begin{align}
  \label{ns-condition1}
  \tr_A E_{AB} = \tr_A F_{AB} = \1_B, &\quad ({\rm A}\not\rightarrow {\rm B}), \\
  \label{ns-condition2}
  \frac{1}{M}E_{AB} + \left(1-\frac{1}{M}\right)F_{AB} = \sigma_A \ox \1_B, &\quad ({\rm B}\not\rightarrow {\rm A}),
\end{align}
with some state $\sigma$.
In other words, both $E_{AB}$ and $F_{AB}$ are Choi-Jamio\l{}kowski matrices 
of channels from $\cL(B)$ to $\cL(A)$
whose weighted sum is equal to a constant channel mapping quantum states in $B$
to a fixed state $\sigma_A$.

Let us first have a closer look at the the mathematical structure of $\O$.
Rewrite $\O$ into the following form,
$$\O_{M_aM_bAB} 
  = \frac{1}{M}D_{M_aM_b}\ox E_{AB} + \left(1-\frac{1}{M}\right)\tilde{D}_{M_aM_b} \ox F_{AB},$$
where $\tilde{D}_{M_aM_b}=\frac{1}{M-1}(\1-D)_{M_aM_b}$ is the Choi-Jamio\l{}kowski matrix of the 
classical channel $\widetilde{\cI}_M$ that sends each $m$ into a uniform 
distribution of $m'\neq m$. In other words, the no-signalling correlation 
$\Pi$ should have the following form:
\begin{equation}\label{simplified-ns}
  \Pi=\frac{1}{M}\cI_M\ox\cE+\left(1-\frac{1}{M}\right)\widetilde{\cI}_M\ox \cF,
\end{equation}
and $\cE$ and $\cF$ are CP maps from $B$ to $A$ corresponding to Choi-Jamio\l{}kowski 
matrices $E_{AB}$ and $F_{AB}$, respectively. 

Now we can directly apply Lemma \ref{rank2-ns} to $\Pi$ in Eq. (\ref{simplified-ns}) 
to obtain the no-signalling constraints in Eqs. (\ref{ns-condition1}) and (\ref{ns-condition2}). 
First, the constraints for ${\rm A}\not\rightarrow {\rm B}$ are automatically satisfied due 
to the special forms of $\cI_M$ and $\widetilde{\cI}_M$, and $p=1/M$ is the unique 
number such that $\frac1M \cI_M+\left(1-\frac1M\right)\widetilde{\cI}_M$ is constant. 
The no-signalling constraint ${\rm B}\not\rightarrow {\rm A}$ is equivalent to that $\cE$ and $\cF$ 
are CPTP maps, and  $\frac1M \cE + \left(1-\frac1M\right)\cF$ is some constant channel. 

When composing $\cN$ with $\Pi$ in Eq. (\ref{simplified-ns}), we have the following channel
\[
  \cM = \frac{1}{M}\cI_M\tr(\cN\circ\cE) 
         + \left(1-\frac{1}{M}\right)\widetilde{\cI}_M \tr(\cN\circ\cF).
\]
The zero-error constraint requires $\cM=\cI_M$, which is equivalent to $\tr(\cN\circ \cF)=0$, or 
$$\tr F_{AB}J_{AB}=0,$$
where $J_{AB}=(\cN^\dag \ox \id_B)(\Phi_{B'B})$ is the Choi-Jamio\l{}kowski matrix 
of $\cN$ (strictly speaking, this is the complex conjugate of the Choi-Jamio\l{}kowski 
matrix of $\cN$, but this does not make any difference to the problem we are studying).
  
Since $E_{AB}$ and $F_{AB}$ depend on each other, we can eliminate one of them. 
For instance, we only keep $E_{AB}$. Then the existence of $F_{AB}$ will be equivalent to 
$$\tr_A E_{AB}=\1_B, 0\leq E_{AB}\leq M\sigma_A\ox \1_B, \tr P_{AB}(M\sigma_A\ox \1_B-E_{AB})=0.$$
By absorbing $M$ into $\sigma_A$ and introducing $S_A=M\sigma_A$, we get $M$ as 
the integer part of
\begin{equation}\label{oneshot-cns0}
  \max \tr S_A \text{ s.t. } 0\leq E_{AB}\leq S_A\ox \1_B, 
                             \tr_A E_{AB}=\1_B,\ \tr P_{AB}(S_A\ox \1_B-E_{AB})=0,
\end{equation}
which ends the proof.
\end{proof} 

Now we provide a detailed derivation of the form (\ref{omega-cns0}) of the 
no-signalling correlation. Assume the channel is $\cN: \cL(A)\rightarrow \cL(B'')$, and 
the no-signalling correlation we will use is 
$\Pi: \cL(M_a'\ox B')\rightarrow \cL(A\ox M_b)$.  
Suppose we can send $M$ messages exactly by one use of the channel 
$\cN$ when assisted by $\Pi$. Then we should have (refer to Fig. 3 and note $A_i=M_a'$, $A_o=A$, $B=B''$, $B_i'=B$, and $B_o=M_b$)
\begin{equation}\label{teleportation}
\tr_{BB''}[\Phi_{BB''} ((\cN\ox \id_{M_b B})\circ (\Pi\ox \id_{B}))(\ketbra{m}{m}_{M_a'}\ox \Phi_{B'B})]=\ketbra{m}{m}_{M_b}\ \forall m\in\{1,\ldots,M\}.
\end{equation}
The Choi-Jamio\l{}kowski matrix of $\Pi$, say $\O_{M_aM_bAB}$, should satisfy the following constraints
$$\O_{M_aM_bAB}\geq 0,\ \tr_{AM_b} \O_{M_aM_bAB}=\1_{M_aB}.$$
Furthermore, noticing that
$$\ketbra{m}{m}_{M_a'}=\tr_{M_a}(\Phi_{M_a'M_a}(\1_{M_a'}\ox\ketbra{m}{m}_{M_a}^{T}))
      \ \text{and}\ 
  \ketbra{m}{m}=\ketbra{m}{m}^T,$$
we have
$$(\Pi\ox \id_{B})(\ketbra{m}{m}_{M_a'}\ox \Phi_{B'B})
  =\tr_{M_a} (\O_{M_aM_bAB}\cdot (\ketbra{m}{m}_{M_a}\ox \1_{M_bAB})).$$
Thus the left hand side (l.h.s) of Eq. (\ref{teleportation}) gives, for all
$m\in\{1,\ldots,M\}$,
$$\tr_{BB'' M_a} [((\cN\ox \id_{M_aM_bB})\O_{M_aM_bAB})
                  (\Phi_{BB''}\ox \ketbra{m}{m}_{M_a}\ox \1_{M_b})]=\ketbra{m}{m}_{M_b}.$$
In other words, for any $m\neq m'$, we have
$$\tr \bigl( ((\cN\ox \id_{M_aM_bB})\O_{M_aM_bAB})(\Phi_{BB''}\ox\ketbra{m}{m}_{M_a}\ox \ketbra{m'}{m'}_{M_b}) \bigr)
  =0.$$

The next crucial step is to simplify the form of $\O$. 
We will study all the possible forms of  $\O$ satisfying the 
above equation. (Any such operator is said to be feasible). 
Since both $M_a$ and $M_b$ are classical registers, we can 
perform the dephasing operation on them and assume $\O$ has 
the following form
$$\O=\sum_{m=1}^M\sum_{m'=1}^M\ketbra{mm'}{mm'}_{M_aM_b}\ox \O_{AB}^{(mm')},$$
where $\O_{AB}^{(mm')}$ might not be identical for different pairs $(m,m')$. 

To further simplify the form of $\O$, we next exploit the permutation 
invariance of $\cI_M$. More precisely, 
for any $M\times M$ permutation $\tau\in S_M$, if $\O_{M_aM_bAB}$ is feasible, then
$$\O_{M_aM_bAB}'
  =(\tau_{M_a}\ox \tau_{M_b}\ox \1_{AB})\O_{M_aM_bAB}(\tau_{M_a}\ox \tau_{M_b}\ox \1_{AB})^\dag$$
is also feasible. Furthermore, if $\O'$ and $\O''$ are feasible, 
so is any convex combination $\lambda\O'+(1-\lambda)\O''$ for $0\leq \lambda\leq 1$.
With these two observations, from any feasible $\O$ that has been dephased, 
we can always construct a new feasible $\tilde{\O}$ by performing the 
following twirling operation
$$\tilde{\O}_{M_aM_bAB}
     =\frac{1}{M!}\sum_{\tau\in S_M}
      (\tau_{M_a}\ox \tau_{M_b}\ox \1_{AB})\O_{M_aM_bAB}(\tau_{M_a}\ox \tau_{M_b}\ox \1_{AB})^\dag.$$
By applying Schur's Lemma to the symmetry group $\{\tau_{M_a}\otimes \tau_{M_b}:\tau\in S_M\}$, 
we can see that $\tilde{\O}$ can be chosen as the following form
$$\tilde{\O}_{M_aM_b AB}=\sum_{m} \ketbra{mm}{mm}_{M_aM_b}\ox \O_{AB}^{=}
                          +\sum_{m\neq m'}\ketbra{mm'}{mm'}\ox \O_{AB}^{\neq},$$
which has exactly  the same form as Eq. (\ref{omega-cns0}).

\medskip
The SDP characterization of Theorem~\ref{capacity-1-shot},
\begin{equation}\begin{split}
  \label{eq:Upsilon}
  \U(K) &= \max \tr S_A \ \text{ s.t. }\  0 \leq E_{AB} \leq S_A \ox \1_B, \\
        &\phantom{= \max \tr S_A \text{ s.t. }} \tr_A E_{AB} = \1_B, \\
        &\phantom{= \max \tr S_A \text{ s.t. }} \tr P_{AB}(S_A\ox\1_B-E_{AB}) = 0,
\end{split}\end{equation}
has a dual. It is given by
\begin{equation}\begin{split}
  \label{eq:Upsilon-dual}
  \U(K) &= \min \tr T_B \ \text{ s.t. }\  z \geq 0,\ W_{AB} \geq 0, \\
        &\phantom{= \min \tr T_B \text{ s.t. }} z P_{AB} - W_{AB} \leq \1_A \ox T_B, \\
        &\phantom{= \min \tr T_B \text{ s.t. }} \tr_B (z P - W)_{AB} = \1_A     \\
        &= \min \tr T_B \ \text{ s.t. }\  F_{AB} \leq \1_A \ox T_B,  \\
        &\phantom{= \min \tr T_B \text{ s.t. }} \tr_B F_{AB} = \1_A, \\
        &\phantom{= \min \tr T_B \text{ s.t. }} (\1-P)_{AB}F_{AB}(\1-P)_{AB} \leq 0,
\end{split}\end{equation}
which can be derived by the usual means; by strong duality, which applies
here, the values of both the primal and the dual SDP coincide.

\medskip
For an easy example of the evaluation of these SDPs, $\U(\Delta_\ell) = \ell$
for the cq-graph $\Delta_\ell$ of the noiseless classical channel $\cI_\ell$
of $\ell$ symbols. The Kraus operator space of the noiseless $\ell$-level quantum channel
$\id_\ell$ is $\CC\1$, and has $\U(\CC\1)=\ell^2$.

\medskip
\begin{remark}
  By inspection of the primal SDP, for any $K_1$ and $K_2$,
  $\U(K_1\ox K_2) \geq \U(K_1) \U(K_2)$,
  because the tensor product of feasible solutions of the SDP (\ref{eq:Upsilon})
  is feasible for $K_1\ox K_2$. In particular,
  $\U(K\ox\Delta_\ell) \geq \ell\,\U(K)$. We do not know whether equality
  holds here (sometimes it does); the most natural way for proving this
  would be to use the dual SDP (\ref{eq:Upsilon-dual}), but to use it to
  show ``$\leq$'' by tensoring together dual feasible solutions would
  require that $T\geq 0$.

  Leaving this aside, this last observation is why $\U(K)$ is rightfully called 
  the no-signalling assisted independence number, and not its integer part. 
  Indeed, the number of messages we can send via $K\ox\Delta_\ell$ is 
  $\lfloor \U(K\ox\Delta_\ell) \rfloor > \ell \lfloor \U(K) \rfloor$,
  for non-integer $\U(K)$ and sufficiently large $\ell$.
\end{remark}

\subsection{Zero-error assisted simulation cost}
\label{subsec:1-shot-cost}
For convenience we restate here the two theorems already announced in the 
introduction.

\medskip\noindent
{\bf Theorem~\ref{review-simulation-1-shot}}\ 
\textit{The one-shot zero-error classical simulation cost (quantified as the minimum
  number of messages) of a quantum channel 
  $\cN$ under quantum no-signalling assistance is given by 
    $\lceil 2^{-H_{\min}(A|B)_{J}}\rceil$. 
  Here, Here, $J_{AB}=(\id_A\ox \cN)\Phi_{AA'}$ is the Choi-Jamio\l{}kowski  matrix of $\cN$, and $H_{\min}(A|B)_{J}$ is the conditional min-entropy 
  defined as follows \cite{KRS2009,Tomamichel-PhD}:
  \[
    2^{-H_{\min}(A|B)_{J}} = \S(\cN)= \min \tr T_B,\ {\rm s.t. }\ J_{AB}\leq \1_A\ox T_B.
  \]
}   

\medskip\noindent
{\bf Theorem~\ref{thm:L-1-shot-cost}}\ 
\textit{The one-shot zero-error classical simulation cost of a Kraus operator space $K$ under quantum no-signalling assistance
  is given by the integer ceiling of  
  $$\S(K)=\min \tr T_B \text{ s.t. } 0\leq F_{AB}\leq \1_A\ox T_B,\ 
                                     \tr_B F_{AB}=\1_A,\ \tr F_{AB}(\1_{AB}-P_{AB})=0,$$
  where $P_{AB}$ denotes the projection onto the subspace $(\1\ox K)\ket{\Phi}$. }

\medskip
\begin{proof}
In this case we have $A_o = M_a$ and $B_i = M_b$ are classical, $A_i = A$, $B_o = B$. We will show that w.l.o.g.
\begin{equation}\label{simulation:ns}
 \O_{ABM_aM_b} = \frac{1}{M}E_{AB}\ox D_{M_aM_b} + \left(1-\frac{1}{M}\right)F_{AB}\ox \tilde{D}_{M_aM_b},
\end{equation}
with positive semidefinite $E_{AB}$ and $F_{AB}$. 
Then according to Lemma \ref{rank2-ns}, the no-signalling conditions are equivalent to
\begin{align*}
  \tr_B E_{AB} = \tr_B F_{AB} &= \1_A, \\
  \frac{1}{M} E_{AB} + \left(1-\frac{1}{M}\right) F_{AB} &= \1_A \ox \gamma_B,
\end{align*}
for a state $\gamma$.

Now, there are two variants of the problem. First, to simulate the precise channel 
$\cN$, $E_{AB}$ has to be equal to $J_{AB}$. By identifying $T_B=M\gamma_B$ and eliminating 
$F_{AB}$, we get the solution for the minimal $M$ as smallest integer 
\[
  \geq 2^{-H_{\min}(A|B)_{J}} = \min \tr T_B \text{ s.t. } J_{AB}\leq \1_A\ox T_B,
\]
proving Theorem~\ref{review-simulation-1-shot}.
The latter is then also asymptotic cost of simulating many copies of $N$ because 
the conditional min-entropy is additive, Eq. (\ref{eq:channel-sim}).

Furthermore, to simulate the ``cheapest'' $\cN$ with Choi-Jamio\l{}kowski 
matrix supporting on $P_{AB}$, we get
\[
  \min \tr T_B \text{ s.t. } 0 \leq F_{AB} \leq \1_A \ox T_B,\ 
                             \tr_B F_{AB} = \1_A,\ \tr F_{AB}(\1_{AB}-P_{AB}) = 0,
\]
which is the claim of Theorem~\ref{thm:L-1-shot-cost}.
\end{proof}

Now we provide a more detailed derivation of the form 
of no-signalling correlations. Suppose we can use a noiseless 
classical channel $\cI_M$ to simulate a quantum channel 
$\cN:\cL(A')\rightarrow \cL(B)$. The no-signalling correlation we 
will use is $\Pi: \cL(A'\ox M_b')\rightarrow \cL(M_a\ox B)$. 
Then we have
$$\cN(\rho_{A'})
  =\sum_{m=1}^M \tr_{M_a}[(\ketbra{m}{m}_{M_a}\ox \1_{B_o})\Pi(\rho_{A'}\otimes \ketbra{m}{m}_{M_b'})].$$
Denote the Choi-Jamio\l{}kowski matrix of $\Pi$ as 
\[
  \O_{M_aM_bAB}=(\id_{M_bA}\ox \Pi)(\Phi_{M_bM_b'}\ox \Phi_{AA'}).
\]
Thus the Choi-Jamio\l{}kowski matrix of $\cN$ is given by
\begin{align*}
  J_{AB} =  (\id_{A}\otimes \cN)(\Phi_{AA'})
         &= \sum_{m}\tr_{M_a}[\ketbra{m}{m}_{M_a} (\id\ox\Pi)(\Phi_{AA'}\ox\ketbra{m}{m}_{M_b'})] \\
         &= \sum_m \tr_{M_aM_b}[\ketbra{mm}{mm}_{M_aM_b}\O_{M_aM_bAB}] \\
         &= \tr_{M_aM_b} D_{M_aM_b}\O_{M_aM_bAB},
\end{align*}
where $D_{M_aM_b}=\sum_{m} \ketbra{mm}{mm}_{M_aM_b}$ is the Choi-Jamio\l{}kowski
matrix of the noiseless classical channel ${\cI}_M$, as before.

In summary, to simulate $\cN$ exactly, we have
\[
  J_{AB}=\tr_{M_aM_b}D_{M_aM_b} \O_{ABM_aM_b} \text{ s.t. } \O\ge 0,\ \tr_{M_aB} \O = \1_{AM_b},
\]
and the no-signalling constraints on $\O$.

By dephasing and twirling classical registers $M_a$ and $M_b$ (refer to the case of zero-error communication), we can choose w.l.o.g  $\O_{M_aM_bAB}$ to have the form in Eq. (\ref{simulation:ns}).

\medskip
We end this subsection, like the previous one, recording the primal and dual
SDP form of $\S(K)$, again equal by strong duality:
\begin{equation}\begin{split}
  \label{eq:Sigma}
  \S(K) &= \min \tr T_B \ \text{ s.t. }\  0 \leq F_{AB} \leq \1_A \ox T_B, \\
        &\phantom{= \min \tr T_B \text{ s.t. }} \tr_B F_{AB} = \1_A, \\
        &\phantom{= \min \tr T_B \text{ s.t. }} \tr (\1-P)_{AB}F_{AB} = 0.
\end{split}\end{equation}
Its dual SDP is
\begin{equation}\begin{split}
  \label{eq:Sigma-dual}
  \S(K) &= \max \tr S_A \ \text{ s.t. }\  0 \leq E_{AB},\ \tr_A E_{AB} = \1_B, \\
        &\phantom{= \max \tr S_A \text{ s.t. }} P_{AB}(S_A \ox \1_B - E_{AB})P_{AB} \leq 0.
\end{split}\end{equation}

\medskip
\begin{remark}
Just as for the no-signalling assisted independence number $\U(K)$,
$\S(\Delta_\ell) = \ell$ and hence $\Sigma(K\ox\Delta_\ell) \leq \ell\,\Sigma(K)$.
For cq-graphs $K$ we can say a bit more: since $\Delta_\ell$ is a cq-graph, too,
we find $\Sigma(K\ox\Delta_\ell) = \ell\,\Sigma(K)$, by 
Proposition~\ref{prop:Sigma-cq-multi} in Section~\ref{subsec:asymptotic-cost-cq} below.
\end{remark}

\section{Towards asymptotic capacity and cost}
\label{sec:asymptotics}
For a classical channel with bipartite graph $\Gamma$, such that
\[
  K = \operatorname{span}\{ \ket{y}\!\bra{x} : (x,y) \text{ edge in } \Gamma\}
\]
is a special type of cq-graph, it was shown in~\cite{CLMW2011} that 
\[
  C_{0,\NS}(K) = S_{0,\NS}(K) = \log \Aram(K) = \log \alpha^*(\Gamma),
\]
where $\alpha^*(\Gamma)$ is the \emph{fractional packing number}
\cite{Shannon1956}
(which is equal to its \emph{fractional covering number}):
\begin{equation}\begin{split}
  \label{eq:fractional-pack-cov}
  \alpha^*(\Gamma) 
           &= \max \sum_x p_x \ \text{ s.t. }\ \  \sum_x p_x \Gamma(y|x) \leq 1~\forall y,
                                             \ 0 \leq p_x \leq 1~\forall x \ , \\
           &= \min \sum_j q_j \ \text{ s.t. }\ \  \sum_y q_y \Gamma(y|x) \geq 1~\forall x,
                                             \ \ q_y \geq 1~\forall y .
\end{split}\end{equation}
In fact, there it was shown (and one can check immediately from
Theorems~\ref{capacity-1-shot} and \ref{thm:L-1-shot-cost}) 
that in this case
\[
  \U(K) = \S(K) = \Aram(K) = \alpha^*(\Gamma).
\]
Furthermore, to attain the simulation cost $S_{0,\NS}(\cN)$, as well
as $S_{0,\NS}(K)$, asymptotically no non-local resources beyond
shared randomness are necessary \cite{CLMW2011}.
We will see in the following that quantum channels exhibit more complexity.
Indeed, while $2^{-H_{\min}(A|B)_J}$ is clearly multiplicative in the channel (or
equivalently in $J$), cf.~\cite{KRS2009,Tomamichel-PhD}, 
and $\alpha^*(\Gamma)$ is well-known to be multiplicative under 
direct products of graphs, cf.~\cite{CLMW2011},
by contrast $\U$ and $\S$ are only super- and sub-multiplicative, respectively:
\begin{align*}
  \U(K_1 \ox K_2) &\geq \U(K_1)\U(K_2), \\
  \S(K_1 \ox K_2) &\leq \S(K_1)\S(K_2).
\end{align*}
We know that the first inequality can be strict 
(see subsection \ref{2-state-example} below),
and suspect that the second can be strict, too. Thus we are facing a 
regularization issue to compute the zero-error capacity and the zero-error
simulation cost, assisted by no-signalling correlations:
\begin{align*}
  C_{0,\NS}(K) &= \lim_{n\rightarrow\infty} \frac{1}{n}\log \U(K^{\ox n}) 
                = \sup_n \frac{1}{n}\log \U(K^{\ox n}), \\
  S_{0,\NS}(K) &= \lim_{n\rightarrow\infty} \frac{1}{n}\log \S(K^{\ox n}) 
                = \inf_n \frac{1}{n}\log \S(K^{\ox n}).
\end{align*}
A standard causality argument, together with entanglement-assisted
coding for the simulation of quantum channels \cite{BSST2003}, implies also
\begin{equation}
  C_{0,\NS}(K) \leq C_{\min {\rm E}}(K) \leq S_{0,\NS}(K),
\end{equation}
where $C_{\min{\rm E}}(K)$ is the minimum of the entanglement-assisted classical 
capacity of quantum channels $\cN$ such that $K(\cN)<K$, i.e., $K(\cN)$ is a subspace of $K$. With the
quantum mutual information $I(\rho;\cN) := I(A:B)_\sigma$ of the state
$\sigma^{AB} = (\id\ox\cN)\phi_\rho$, where $\phi_\rho$ is a purification
of $\rho$:
\begin{equation}\begin{split}
  \label{eq:C-min-E}
  C_{\min {\rm E}}(K) &= \min C_{\rm E}(\cN) \text{ s.t. } K(\cN) < K                          \\
       &= \min_{\cN \text{ s.t.} \atop \cK({\cal N}) < K} \max_\rho \ \ \, I(\rho;{\cal N}) \\
       &= \ \ \, \max_\rho \min_{\cN \text{ s.t.} \atop \cK({\cal N}) < K} I(\rho;{\cal N}),
\end{split}\end{equation}
where the last equality follows from the Sion's minimax theorem \cite{Sion:minimax}. For more properties, including the above minimax formulas,
and an operational interpretation of $C_{\min {\rm E}}(K)$
as the entanglement-assisted capacity of $K$ (i.e.~the maximum rate of block 
codings adapted simultaneously to all channels $\cN^{(n)}$ such that
$K(\cN^{(n)}) < K^{\ox n}$), we refer the reader to \cite{DSW2013}.

To put better and easier to use bounds on $C_{0,\NS}$ and $S_{0,\NS}$,
we introduce the semidefinite packing number:
\begin{equation}\begin{split}
  \label{eq:Aram}
  \Aram(K) &= \max \tr S_A \ \text{ s.t. }\  0 \leq S_A,\ \tr_A P_{AB}(S_A\ox\1_B) \leq \1_B \\
           &= \min \tr T_B \ \text{ s.t. }\  0 \leq T_B,\ \tr_B P_{AB}(\1_A\ox T_B) \geq \1_A,
\end{split}\end{equation}
which we have given in primal and dual form; this generalizes the
form given in Eq.~(\ref{eq:Aram-cq}) for cq-graphs. It
was suggested to us in the past by Aram Harrow \cite{Harrow2010} for
its nice mathematical properties.

A slightly modified and more symmetric form is given by
\begin{equation}\begin{split}
  \label{eq:revised-Aram}
  \widetilde{\Aram}(K) &= \max \tr S_A \ \text{ s.t. }\ 0 \leq S_A,\ 
                                                    \tr_A P_{AB}(S_A\ox\1_B)P_{AB} \leq \1_B \\
           &= \min \tr T_B \ \text{ s.t. }\  0 \leq T_B,\ \tr_B P_{AB}(\1_A\ox T_B)P_{AB} \geq \1_A,
\end{split}\end{equation}
again both in primal and dual form. 

From these it is straightforward to see that $\Aram(K)$ and 
$\widetilde{\Aram}(K)$ are both sub- and super-multiplicative, and so
\begin{align*}
  \Aram(K_1\ox K_2)             &= \Aram(K_1) \Aram(K_2), \\
  \widetilde{\Aram}(K_1\ox K_2) &= \widetilde{\Aram}(K_1) \widetilde{\Aram}(K_2).
\end{align*}

Both definitions reduce to the familiar notion of fractional packing number 
$\alpha^*(\Gamma)$ when $K$ is associated to a bipartite graph $\Gamma$,
coming from a classical channel $\cN$.
Furthermore, $\Aram(K)$ and $\widetilde{\Aram}(K)$ are equal for cq-graphs. 
However, we will see later that in general they are different.

\subsection{Revised semidefinite packing number and simulation} 
\label{subsec:lower-bound-on-sim}

\begin{proposition}
  For any non-commutative bipartite graph $K < \cL(A \rightarrow B)$,
  \[
    \S(K) \geq \widetilde{\Aram}(K).
  \]
  Consequently, $S_{0,{\rm NS}}(K) \geq \log \widetilde{\Aram}(K)$.
\end{proposition}
\begin{proof}
From the SDP (\ref{eq:Sigma}) for $\S(K)$, we get operators 
$T_B$ and $F_{AB}$ such that $0\leq F_{AB} \leq \1_A\ox T_B$, with $\tr_B F_{AB} = \1_A$
and $F_{AB}(\1-P)_{AB}=0$. Hence $P_{AB}(\1_A\ox T_B)P_{AB} \geq P_{AB}F_{AB}P_{AB} = F_{AB}$ and so
\[
  \tr_B P_{AB}(\1_A\ox T_B)P_{AB} \geq \tr_B F_{AB} = \1_A,
\]
i.e.~$T$ is feasible for the dual formulation of $\widetilde{\Aram}(K)$, hence the
result follows.
\end{proof}

Below we will see that this bound (both in its one-shot and regularized form)
is in general strict, indeed already for cq-channels it is not a strict equality.

\subsection{Asymptotic assisted zero-error capacity of cq-graphs} 
\label{subsec:asymptotic-capacity-cq}
We do not know whether $\U(K)$ is in general related to $\Aram(K)$,
but we will show bounds in either direction for cq-channels. 

Suppose that the cq-channel $\cN$ acts as $\proj{i} \longmapsto \rho_i$, with
support $K_i$ and support projection $P_i$ of $\rho_i$.  
Then, the non-commutative bipartite graph $K$ associated with $\cN$ is given by
$K = \sum_i \ket{i} \ox K_i$, and the projection for the Choi-Jamio\l{}kowski 
matrix is $P = \sum_i \proj{i}^A \ox P_i^B$. The SDP (\ref{eq:Upsilon}) easily simplifies to
\begin{equation}\begin{split}
  \label{eq:Upsilon:cq-channel}
  \U(K) &= \max \sum_i s_i \ \text{ s.t. }\  0 \leq s_i,\ 0 \leq R_i \leq s_i(\1-P_i),\\
        &\phantom{= \max \sum_i s_i \ \text{ s.t. }} \sum_i (s_i P_i + R_i) = \1.
\end{split}\end{equation}
The semidefinite packing number (\ref{eq:Aram}), on the other hand, simplifies to
\begin{equation}
  \label{eq:Aram:cq-channel}
  \Aram(K) = \max \sum_i s_i \ \text{ s.t. }\  0 \leq s_i,\ \sum_i s_i P_i \leq \1.
\end{equation}

As this is an SDP relaxation of the problem (\ref{eq:Upsilon:cq-channel}),
we obtain:
\begin{lemma}
  \label{lemma:U-Aram-bound}
  For a non-commutative bipartite cq-graph  $K < \cL(A \rightarrow B)$,
  \[
    \U(K) \leq \Aram(K).
  \]
  Consequently, $C_{0,{\rm NS}}(K) \leq \log \Aram(K)$.
  \qed
\end{lemma}

In general, $\U(K)$ can be strictly smaller than $\Aram(K)$, see Subsection \ref{2-state-example}
below, but we shall prove equality for the regularization, by exhibiting a lower bound
\[
  \U(K^{\ox n}) \geq n^{-O(1)} \Aram(K^{\ox n})
                =    n^{-O(1)} \Aram(K)^n.
\]
The way to do this is to take a feasible solution $s_i$ of the SDP
(\ref{eq:Aram:cq-channel}); w.l.o.g.~its value $\sum_i s_i > 1$, otherwise 
the above statement is trivial. On strings $\underline{i} = i_1\ldots i_n$ of 
length  $n$ this gives a feasible solution 
$s_{\underline{i}} = s_{i_1}\cdots s_{i_n}$ for $K^{\ox n}$, with value 
\[
  \sum_{\underline{i}} s_{\underline{i}} = \left( \sum_i s_i \right)^n.
\]
Now for $a$ different symbols $i=1,\ldots,a$ there are at most $(n+1)^a$ many
types of strings, hence there is one type $\tau$ such that
\[
  \sum_{\underline{i}\in\tau} s_{\underline{i}} \geq (n+1)^{-a} \left( \sum_i s_i \right)^n.
\]
Restricting the input of the channel to $\underline{i}\in\tau$ (while the
output is still $B^n$), we thus loose only a polynomial factor of the
semidefinite packing number. What we gain is that the inputs are all of
the same type, which means that the output projectors 
$P_{\underline{i}} = P_{i_1}\ox\cdots\ox P_{i_n}$ are related to each other
by unitaries $U^\pi$ permuting the $n$ $B$-systems. Note that then also
all of the $s_{\underline{i}}$ on the left hand side above are the same,
say $s_\tau$, and hence the left hand side is $s_\tau\,|\tau|$.

\medskip
Abstractly, we are thus in the following situation: 
Assume that there exists a transitive
group action by unitary conjugation on the $P_i$, i.e.~we have a finite
group $G$ acting transitively on the labels $i$ (running over a set of size $N$), 
and a unitary representation
$U^g$, such that $P_{i^g} = (U^g)^\dagger P_i U^g$ for $g\in G$. 
In other words, the entire set of $\{P_i\}$ is the orbit 
$\left\{(U^g)^\dagger P_0 U^g\right\}$ of a fiducial element $P_0$
under the group action.
Then we can twirl the SDP (\ref{eq:Aram:cq-channel}) and w.l.o.g.~assume that
all $s_i$ are identical to $s$, so the constraint reduces to
$s\sum_i P_i \leq \1$, meaning that the largest admissible $s$ is
$\left\| \sum_i P_i \right\|_{\infty}^{-1}$, and the semidefinite packing number
$\Aram(K) = sN = \left\| \frac{1}{N} \sum_i P_i \right\|_{\infty}^{-1}$.

From this we see that the representation theory of $U^g$ has a bearing
on the semidefinite packing number $\Aram(K)$; cf.~\cite{FH1991}
for some basic facts that we shall invoke in the following.
Indeed, it also governs 
$\U(K)$ since we can do the same twirling operation and find that in
the SDP (\ref{eq:Upsilon:cq-channel}) we have w.l.o.g.~that all $s_i$
are equal to the same $s$, but also that for $P_j = (U^g)^\dagger P_i U^g$,
w.l.o.g.~$R_j = (U^g)^\dagger R_i U^g$. In particular, in this case
\begin{equation}\begin{split}
  \label{eq:Upsilon:covariant-cq-channel}
  \U(K) &= \max sN \ \text{ s.t. }\  0 \leq s,\ 0 \leq R_0 \leq s(\1-P_0),\\
        &\phantom{= \max sN \text{ s.t. }}
                     \frac{1}{|G|} \sum_g (U^g)^\dagger (s P_0 + R_0) U^g = \frac{1}{N} \1.
\end{split}\end{equation}
Let us be a little more explicit in the reduction of the SDP (\ref{eq:Upsilon:cq-channel})
to the above SDP: Let $s$ and $R_0$ be feasible as above, and denote by $G_0$
the subgroup of $G$ leaving $0$ invariant, $G_0 = \{g \in G: 0^g=0 \}$. By Lagrange's
Theorem, $N=|G/G_0|$. Then,
\[
  \overline{R}_0 := \frac{1}{|G_0|} \sum_{g\in G_0} (U^g)^\dagger R_0 U^g
\]
is also feasible with the same $s$, using $(U^g)^\dagger P_0 U^g = P_0$ for
all $g\in G_0$. Letting $\overline{R}_i := (U^g)^\dagger \overline{R}_0 U^g$
for any $g$ such that $i=0^g$, and $s_i=s$, then yields a feasible solution for
(\ref{eq:Upsilon:cq-channel}) -- and this is well-defined. Thus $\U(K)$ is not 
smaller than the above SDP.
In the other direction, let $s_i$ and $R_i$ be feasible for (\ref{eq:Upsilon:cq-channel}),
i.e.~$0 \leq R_i \leq s_i(\1-P_i)$. 
Letting $\overline{s} = \frac{1}{N}\sum_i s_i$ and
\[
  \overline{R}_0 := \frac{1}{G} \sum_{g\in G} U^g R_{0^g} (U^g)^\dagger
\]
yields a feasible solution for (\ref{eq:Upsilon:covariant-cq-channel}).

If the representation $U^g$ happens to be irreducible, we are lucky because then
the group average in the second line in Eq.~(\ref{eq:Upsilon:covariant-cq-channel})
is automatically proportional to the identity, by Schur's Lemma. 
Hence the optimal choice is $R_0 = R_i = 0$ and
we find $\U(K) = \Aram(K)$. In general this won't be the case, but if the 
representation $U^g$ is ``not too far'' from being irreducible, in a sense made
precise in the following proposition, then $\U(K)$ is not too much smaller
than $\Aram(K)$:

\begin{proposition}
  \label{key-lemma}
  For a set of projections $P_i$ on $B$ with a transitive group action by
  conjugation under $U^g$, let 
  \[
    B = \bigoplus_\lambda \cQ_\lambda \ox \cR_\lambda
  \]
  be the isotypical decomposition of $B$ into irreps $\cQ_\lambda$ of
  $U^g$, with multiplicity spaces $\cR_\lambda$. Denote the number of
  terms $\lambda$ by $L$, and the largest occurring multiplicity by
  $M = \max_\lambda |\cR_\lambda|$. Then, for the corresponding cq-graph $K$,
  \[
    \U(K) \geq \frac{1}{4 L^2 M^{9/2}} \Aram(K),
  \]
  if $\Aram(K) \geq 64 L^6 M^{14}$.
\end{proposition}
\begin{proof}
Assume that we have a feasible $s^*$ for $\Aram(K)$ such that 
$s^*N \geq 64 L^6 M^{14}$.

Our point of departure is the SDP (\ref{eq:Upsilon:covariant-cq-channel}):
for given $s\geq 0$ and $R_0\geq 0$, Schur's Lemma tells us
\[
  \frac{1}{|G|} \sum_g (U^g)^\dagger (s P_0 + R_0) U^g 
         = \frac{1}{N} \sum_\lambda Q_\lambda \ox \zeta_\lambda,
\]
where $Q_\lambda$ is the projection onto the irrep $\cQ_\lambda$,
$\zeta_\lambda$ is a semidefinite operator on $\cR_\lambda$.
Feasibility of $s$ and $R_0$ (to be precise: the equality constraints)
is equivalent to $\zeta_\lambda = \Pi_\lambda$, the projection onto $\cR_\lambda$,
for all $\lambda$. 

Now, for each $\lambda$ choose an orthogonal basis $\{Z^{(\lambda)}_\mu\}$
of Hermitians over $\cR_\lambda$, 
with $Z^{(\lambda)}_0 = \frac{1}{\tr\Pi_\lambda}\Pi_\lambda$
and $\| Z^{(\lambda)}_\mu \|_2 = 1$ for $\mu \neq 0$. Then the
$\frac{1}{\tr Q_\lambda}Q_\lambda \ox Z^{(\lambda)}_\mu$ form a basis of
the $U^g$-invariant operators, hence our SDP can be rephrased as
\begin{equation}\begin{split}
  \label{eq:Upsilon:covariant-cq-channel-2}
  \U(K) &= \max sN \ \text{ s.t. }\  0 \leq s,\ 0 \leq R_0 \leq s(\1-P_0),\\
        &\phantom{= \max sN \text{ s.t. }}
           \forall\lambda\,\mu\quad
           \tr\,(s P_0 + R_0)\!\left(\frac{Q_\lambda}{\tr Q_\lambda}\ox Z^{(\lambda)}_\mu\right)
                                                                    = \frac{1}{N}\delta_{\mu 0}.
\end{split}\end{equation}
Notice that here, the semidefinite constraints on $R_0$ leave quite some
room, whereas we have ``only'' $LM^2$ linear conditions to satisfy.
Given $s^*$ satisfying the constraint of $\Aram(K)$, our strategy now will be
to show that we can construct a $0 \leq R_0 \leq \frac{2\beta}{N}(\1-P_0)$ such that 
the above equations are true with $s=s^*$ on the left hand side, and
with a factor $\beta$ on the right hand side.
We will choose $\beta = 4 L^2 M^{9/2}$ and thus there is a feasible
solution with $s=s^*/\beta$ to (\ref{eq:Upsilon:covariant-cq-channel-2}),
hence $\U(K) \geq s^*N/\beta$ as claimed.

In detail, introduce a new variable $X\geq 0$, with
\[
  R_0 = \frac{\beta}{N}(\1-P_0)X(\1-P_0),
\]
which makes sure that $R_0$ is automatically supported on the orthogonal complement of $P_0$.
Rewrite the equations
\[
  \tr\,(s^* P_0 + R_0)\!\left(\frac{Q_\lambda}{\tr Q_\lambda}\ox Z^{(\lambda)}_\mu\right)
                                                             = \frac{\beta}{N}\delta_{\mu 0}
\] 
in terms of $X$, introducing the notation
\[
  C_{\lambda\mu} = \frac{1}{\tr Q_\lambda}Q_\lambda \ox Z^{(\lambda)}_\mu,
  \quad
  D_{\lambda\mu} = (\1-P_0) C_{\lambda\mu} (\1-P_0).
\]
This gives, noting $\tr P_0 C_{\lambda\mu} = \tr P_i C_{\lambda\mu}$ for all $i$
because of the $U^g$ invariance of $C_{\lambda\mu}$,
\begin{equation}\begin{split}
  \label{eq:X-and-t}
  \tr X D_{\lambda\mu} &= \delta_{\mu 0} - \frac{1}{\beta}\tr\, s^* N P_0 C_{\lambda\mu} \\
                       &= \delta_{\mu 0} - \frac{1}{\beta}\tr \left(\sum_i s^* P_i\right) C_{\lambda\mu} \\
                       &=: \delta_{\mu 0} - \frac{1}{\beta} t_{\lambda\mu}.
\end{split}\end{equation}
What we need of the coefficients $t_{\lambda\mu}$ is that they cannot be too large:
from $0 \leq \sum_i s^* P_i \leq \1$ we get
\begin{equation}
  \label{eq:t-bound}
  | t_{\lambda\mu} | \leq \| C_{\lambda\mu} \|_1 = \| Z^{(\lambda)}_\mu \|_1 \leq \sqrt{M}.
\end{equation}

Our goal will be to find a ``nice'' dual set $\{\widehat{D}_{\lambda\mu}\}$ to the 
$\{D_{\lambda\mu}\}$, 
i.e.~$\tr D_{\lambda\mu} \widehat{D}_{\lambda'\mu'} = \delta_{\lambda\lambda'}\delta_{\mu\mu'}$,
with which we can write a solution
$X = \sum_{\lambda\mu} \left(\delta_{\mu 0} - \frac{1}{\beta}t_{\lambda\mu}\right) \widehat{D}_{\lambda\mu}$.
To this end, we construct first the dual set $\widehat{C}_{\lambda\mu}$ of the 
$\{C_{\lambda_\mu}\}$, which is easy:
\[
  \widehat{C}_{\lambda\mu} = Q_\lambda \ox \widehat{Z}^{(\lambda)}_\mu
                           = \begin{cases}
                               Q_\lambda \ox \Pi_\lambda       & \text{ for } \mu = 0, \\
                               Q_\lambda \ox Z^{(\lambda)}_\mu & \text{ for } \mu\neq 0,
                             \end{cases}
\]
so that indeed 
$\tr C_{\lambda\mu} \widehat{C}_{\lambda'\mu'} = \delta_{\lambda\lambda'}\delta_{\mu\mu'}$.
Now, consider the $LM^2\times LM^2$-matrix $T$,
\[\begin{split}
  T_{\lambda\mu,\lambda'\mu'} &= \tr D_{\lambda\mu} \widehat{C}_{\lambda'\mu'} \\
                &= \tr (\1-P_0)C_{\lambda\mu}(\1-P_0) \widehat{C}_{\lambda'\mu'} \\
                &= \delta_{\lambda\lambda'}\delta_{\mu\mu'} - \Delta_{\lambda\mu,\lambda'\mu'},
\end{split}\]
with the deviation
\[
  \Delta_{\lambda\mu,\lambda'\mu'} = \tr P_0 C_{\lambda\mu} (\1-P_0) \widehat{C}_{\lambda'\mu'}
                                     + \tr C_{\lambda\mu} P_0 \widehat{C}_{\lambda'\mu'}.
\]
Here, 
\[\begin{split}
  |\Delta_{\lambda\mu,\lambda'\mu'}| &\leq 2 \| P_0 C_{\lambda\mu} \|_1 \|\widehat{C}_{\lambda'\mu'}\|_\infty
                                      \leq 2 \| P_0 C_{\lambda\mu} \|_1 \\
                  &=    2 \bigl\| P_0 |C_{\lambda\mu}| \bigr\|_1 \\
                  &\leq 2\sqrt{\tr P_0 |C_{\lambda\mu}|}\sqrt{\| C_{\lambda\mu}\|_1},
\end{split}\]
using $\|\widehat{C}_{\lambda'\mu'}\|_\infty \leq 1$, the unitary invariance of
the trace norm, and Lemma~\ref{lemma:tracenorm-trace} below. Since
$|C_{\lambda\mu}| = \frac{1}{\tr Q_\lambda} Q_\lambda \ox \bigl| Z^{(\lambda)}_\mu \bigr|$
is invariant under the action of $U^g$, we have 
$\tr P_0 |C_{\lambda\mu}| = \tr P_i |C_{\lambda\mu}|$ for all $i$, and using
$\sum_i s^* P_i \leq \1$ we get
\begin{equation}
  \label{eq:Delta-bound}
  |\Delta_{\lambda\mu,\lambda'\mu'}| \leq 2\sqrt{\frac{1}{s^*N}\|C_{\lambda\mu}\|_1^2}
                                     \leq 2\sqrt{M} (s^*N)^{-1/2}.
\end{equation}
With this we get that
\begin{equation}\begin{split}
  \| T-\1 \|_\infty \leq \| T-\1 \|_2 
                    &=    \sqrt{ \sum_{\lambda\mu\lambda'\mu'} |\Delta_{\lambda\mu,\lambda'\mu'}|^2} \\
                    &\leq \sqrt{ L^2M^4 4 M (s^*N)^{-1}}
                     \leq \frac{1}{\beta},
\end{split}\end{equation}
if $s^*N \geq 4\beta^2 L^2 M^{5}$. Assuming $\beta \geq 2$ (which will be the
case with our later choice), we thus know that $T$ is invertible; in fact,
we have $T = \1 - \Delta$ with $\| \Delta \| \leq \frac{1}{\beta} \leq \frac12$, hence
$T^{-1} = \sum_{k=0}^\infty \Delta^k$ and so
\[
  \left\| T^{-1} - \1 \right\|_{\infty} =    \left\| \sum_{k=1}^\infty \Delta^k \right\|_{\infty}
                               \leq \sum_{k=1}^\infty \| \Delta \|_{\infty}^k
                               =    \frac{1}{\beta-1}
                               \leq \frac{2}{\beta}.
\]
I.e., writing $T^{-1} = \1 + \widetilde{\Delta}_{\lambda\mu,\lambda'\mu'}$ we get
\begin{equation}
  \label{eq:tilde-Delta-bound}
  |\widetilde{\Delta}_{\lambda\mu,\lambda'\mu'}| \leq \| \widetilde{\Delta} \|_\infty \leq \frac{2}{\beta}.
\end{equation}
The invertibility of $T$ implies that there is a dual set to $\{D_{\lambda\mu}\}$ in 
$\operatorname{span}\{\widehat{C}_{\lambda\mu}\}$. Indeed, from the definition of 
$T_{\lambda\mu,\lambda'\mu'}$ and the dual sets,
\begin{align*}
  \widehat{C}_{\lambda'\mu'} 
           &= \sum_{\lambda\mu} T_{\lambda\mu,\lambda'\mu'} \widehat{D}_{\lambda\mu},
            \ \text{ which can be rewritten as} \\
  \widehat{D}_{\lambda\mu} 
           &= \sum_{\lambda'\mu'} (T^{-1})_{\lambda'\mu',\lambda\mu} \widehat{C}_{\lambda'\mu'}.
\end{align*}

Now we can finally write down our candidate solution to Eqs.~(\ref{eq:X-and-t}):
\[\begin{split}
  X &= \sum_{\lambda\mu} \left(\delta_{\mu 0} - \frac{1}{\beta}t_{\lambda\mu}\right)\widehat{D}_{\lambda\mu} \\
    &= \sum_{\lambda\mu} \left(\delta_{\mu 0} - \frac{1}{\beta}t_{\lambda\mu}\right)
               \sum_{\lambda'\mu'} (T^{-1})_{\lambda'\mu',\lambda\mu} \widehat{C}_{\lambda'\mu'} \\
    &= \sum_{\lambda} \widehat{C}_{\lambda 0}                                                    
        -\frac{1}{\beta}\sum_{\lambda\mu} t_{\lambda\mu}
                        \sum_{\lambda'\mu'} (T^{-1})_{\lambda'\mu',\lambda\mu} \widehat{C}_{\lambda'\mu'}
        +\sum_{\lambda\lambda'\mu'} \widetilde{\Delta}_{\lambda'\mu',\lambda 0} \widehat{C}_{\lambda'\mu'} \\
    &= \1 + \text{Rest}.
\end{split}\]
The remainder term ``Rest'' can be bounded as follows:
\[\begin{split}
  \| \text{Rest} \|_\infty &\leq \frac{1}{\beta} \sum_{\lambda\mu\lambda'\mu'} 2\sqrt{M}
                                                 + \sum_{\lambda\lambda'\mu'} \frac{2}{\beta} \\
                           &=    \frac{2}{\beta} \left( L^2 M^{9/2} + L^2M^2 \right)
                            \leq \frac{4}{\beta} L^2 M^{9/2},
\end{split}\]
using Eqs.~(\ref{eq:t-bound}) and (\ref{eq:tilde-Delta-bound}). Thus we find
$\| \text{Rest} \|_\infty \leq 1$ if $\beta \geq 4L^2 M^{9/2}$ and
$s^*N \geq 4\beta^2 L^2 M^{5} \geq 64 L^6 M^{14}$. 
In this case, $0\leq X \leq 2$ and we can wrap things up:
$R_0 := \frac{\beta}{N}(\1-P_0)X(\1-P_0)$ satisfies 
\[
  0 \leq R_0 \leq \frac{2\beta}{N}(1-P_0) \leq s^*(\1-P_0),
\]
as well as 
\[
  \frac{1}{|G|} \sum_{g\in G} (U^g)^\dagger (s^* P_0 + R_0) U^g = \frac{\beta}{N}\1.
\]
I.e., we get a feasible solution 
$\sum_i \left( \frac{s^*}{\beta}P_i + \frac{1}{\beta}R_i \right) = \1$
for $\U(K)$.
\end{proof}

\begin{lemma}
\label{lemma:tracenorm-trace}
Let $\rho$ be a state and $P$ a projection in a Hilbert space $\cH$. Then,
\[
  \tr \rho P \leq \| \rho P \|_1 \leq \sqrt{\tr\rho P}.
\]
More generally, for $X\geq 0$ and a POVM element $0\leq E \leq \1$,
\[
  \tr XE \leq \| XE \|_1 \leq \sqrt{\tr X}\sqrt{\tr XE}.
\]
\end{lemma}
\begin{proof}
We start with the first chain of inequalities. The left hand one follows
directly from the definition of the trace norm. For the right hand one,
choose a purification of $\rho = \tr_{\cH'} \proj{\psi}$ on $\cH \ox \cH'$.
Now, $\|\rho P\|_1 = \tr \sqrt{P\rho^2 P}$ and
$\rho P = \tr_{\cH'} \proj{\psi}(P\ox\1)$. Thus, by the monotonicity of
the trace norm under partial trace,
\[\begin{split}
  \|\rho\Pi\|_1 &\leq \bigl\|\proj{\psi}(P\ox\1)\bigr\|_1      \\
                &=    \tr \sqrt{(P\ox\1)\proj{\psi}(P\ox\1)} \\
                &=    \sqrt{\tr\proj{\psi}(P\ox\1)}            \\
                &=    \sqrt{\tr\rho P}.
\end{split}\]

The second chain is homogenous in $X$, so we may w.l.o.g.~assume that
$\tr X=1$, i.e.~$X=\rho$ is a state. For a general POVM element $E$ there is
an embedding $U$ of the Hilbert space $\cH$ into a larger Hilbert space $\cH_0$
and a projection $P$ in $\cH_0$ such that $E=U^\dagger P U$.
Then, $\tr \rho E \leq \| \rho E \|_1$ as before by the definition of the
trace norm, and using the invariance of the trace number under unitaries
and the first part,
\[
  \| \rho E \|_1 =    \| \rho U^\dagger P U\|_1 
                 =    \| U\rho U^\dagger P \|_1 
                 \leq \sqrt{\tr U\rho U^\dagger P}
                 =    \sqrt{\tr \rho E},
\]
which concludes the proof.
\end{proof}

\medskip
For the permutation action of $S_n$ on $B^n$, the irreps $\lambda$ are labelled
by Young diagrams with at most $b=|B|$ rows, hence $L \leq (n+1)^{b}$, and it is
well-known that $M \leq (n+b)^{\frac12 b^2}$~\cite[Sec.~6.2]{Harrow2005},
\cite{Christandl2006, Hayashi2002}.
Thus the previous proposition yields directly the following result, observing that $L$ and $M$ are polynomially
bounded in $n$, whilst $\Aram(K^{\ox n}) = \Aram(K)^n$ grows exponentially.

\begin{proposition}
  \label{c-q-capacity-aram}
  Let $K$ be a non-commutative bipartite cq-graph with $a=|A|$ inputs 
  and output dimension $b=|B|$. Then for sufficiently large $n$,
  \[
    \U(K^{\ox n}) \geq \frac{1}{4 (n+1)^{a+2b} (n+b)^{9b^2/4}} \Aram(K)^n.
  \]
  Consequently, $C_{0,{\rm NS}}(K) = \log \Aram(K)$.
  \qed
\end{proposition}

\medskip
Lemma~\ref{lemma:U-Aram-bound} and Proposition~\ref{c-q-capacity-aram} together 
prove Theorem~\ref{thm:cq-capacity}.

\begin{corollary}
  \label{c-q-capacity-additivity}
  For any two non-commutative bipartite cq-graphs $K_1$ and $K_2$, 
  $C_{0,{\rm NS}}(K_1\ox K_2)=C_{0,{\rm NS}}(K_1)+C_{0,{\rm NS}}(K_2)$.
  \qed
\end{corollary}

\subsection{Asymptotic assisted zero-error simulation cost of cq-graphs}
\label{subsec:asymptotic-cost-cq}
Here we study the asymptotic zero-error simulation cost of a non-commutative 
bipartite cq-graph $K = \sum_i \ket{i} \ox K_i$, where the subspace
$K_i$ is the support of the projection $P_i$. Thus, $P = \sum_i \proj{i}^A \ox P_i^B$ 
and the SDP (\ref{eq:Sigma}) easily simplifies to
\begin{equation}\begin{split}
  \label{eq:Sigma:cq-channel}
  \S(K) = \min \tr T \ \text{ s.t. }\ T\geq F_i,\ 0 \leq F_i\leq P_i, \  \tr F_i=1.
\end{split}\end{equation}
Similarly, the dual SDP (\ref{eq:Sigma-dual}) simplifies to
\begin{equation}
  \label{eq:dual:cq-channel}
  \S(K) = \max \sum_i s_i \ \text{ s.t. }\ s_iP_i\leq P_iE_iP_i, 0\leq E_i,\ \sum_i E_i=\1.
\end{equation}

\begin{proposition}
\label{prop:Sigma-cq-multi}
For non-commutative bipartite cq-graphs $K$, $\S(K)$ is multiplicative under tensor products, i.e.
\[
  \S(K_1\ox K_2) = \S(K_1)\S(K_2), 
\]
where $K_1$ and $K_2$ are arbitrary non-commutative bipartite cq-graphs.
\end{proposition}
\begin{proof}
The sub-multiplicativity of $\S(K)$ is evident from (\ref{eq:Sigma:cq-channel}). 
We will show that the super-multiplicativity follows from the dual 
SDP (\ref{eq:dual:cq-channel}). Indeed,  let $K_1$ and $K_2$ correspond to 
$\{P_i\}$ and $\{Q_j\}$, respectively, and assume that $(s_i, E_i)$ and 
$(t_j, F_j)$ are optimal solutions to $\S(K_1)$ and $\S(K_2)$ in dual SDPs, 
respectively. Then we have 
\[
  s_i=\ll_{\min}(P_i E_i P_i),\ \text{with}\ \sum_i E_i=\1_1,\ E_i\ge 0,\ \text{and}\ \S(K_1)=\sum_i s_i.
\]
where $\ll_{\min}(\cdot)$ denotes the minimal eigenvalue of the linear 
operator $P_iE_iP_i$ in the support of $P_i$. Similarly, we have
\[
  t_j=\ll_{\min}(Q_j F_j Q_j),\ \text{with}\ \sum_j F_j=\1_2,\ F_j\ge 0,\ \text{and}\ \S(K_2)=\sum_j t_j.
\]
Clearly, we have 
\[
  s_it_j=\ll_{\min} \bigl( (P_i\ox Q_j)(E_i\ox F_j)(P_i\ox Q_j) \bigr),
          \ \text{and}\ \sum_{ij} E_i\ox F_j=\1_1\ox \1_2,\ \text{and}\ E_i\ox F_j\geq 0.
\]
So $(s_it_j, E_i\ox F_j)$ is a feasible solution to the dual SDP for $\S(K_1\ox K_2)$. 
Since the dual SDP takes maximization, we have
\[
  \S(K_1\ox K_2) \geq \sum_{ij} s_it_j
                 =    \Biggl(\sum_i s_i\Biggr) \! \Biggl(\sum_j t_j\Biggr)
                 =    \S(K_1)\S(K_2).
\]
\end{proof}

From the above result we can read off directly 
\begin{theorem}
  For any non-commutative bipartite cq-graph $K < \cL(A \rightarrow B)$,
  \[
    S_{0,\NS}(K)=\log \S(K).
  \]
  In fact,
  \[
    S_{0,\NS}(K_1\ox K_2)=S_{0,\NS}(K_1)+S_{0,\NS}(K_2),
  \]
  for any two non-commutative bipartite cq-graphs $K_1$ and $K_2$. 
  \qed
\end{theorem}

\medskip
This theorem motivates us to call $\Sigma(K)$ the \emph{semidefinite covering
number}, at least for cq-graphs $K$, in analogy to a result from \cite{CLMW2010}
which states that the zero-error simulation rate of a bipartite
graph is given by its fractional packing number. Note however that while
fractional packing and fractional covering number are dual linear programmes
and yield the same value, the semidefinite versions $\Aram(K) \geq \Sigma(K)$ 
are in general distinct; already in the following Subsection~\ref{2-state-example}
we will see a simple example for strict inequality.

\medskip
For general non-commutative bipartite graph $K$, we do not know whether the 
one-shot simulation cost also gives the asymptotic simulation cost.  
However, this is true when $K={\rm span}\{E_i\}$ corresponds to an 
extremal channel $\cN(\rho)=\sum_i E_i\rho E_i^\dag$, which is well-known
to be equivalent to the set of linear operators $\{E_i^\dag E_j\}$ 
being linearly independent \cite{Choi-extremal}. Actually, 
in this case, there can only be a unique channel $\cN$ such that $K=K(\cN)$, 
and furthermore, a unique channel $\cN^{\otimes n}$ such that 
$K^{\ox n}=K(\cN^{\ox n})$. Hence 
\[
  \S(K^{\ox n})=\S(\cN^{\ox n})=\S(\cN)^{n}=\S(K)^n.
\]
Thus we have the following result:
\begin{theorem}\label{extremal-channel-cost}
  Let $K={\rm span}\{E_i\}$ be an extremal non-commutative bipartite graph in the sense 
  that $\{E_i^\dag E_j\}$ is linearly independent. Then
  \[
    S_{0,{\rm NS}}(K)=\log \S(K) = - H_{\min}(A|B)_J,
  \] 
  for the Choi-Jamio\l{}kowski state of the unique channel $\cN$ with $K(\cN)<K$.
  Furthermore, 
  $$S_{0,{\rm NS}}(K_1\ox K_2)=S_{0,{\rm NS}}(K_1)+S_{0,{\rm NS}}(K_2),$$
  if both $K_1$ and $K_2$ are extremal non-commutative bipartite graphs.
  \qed
\end{theorem}

The fact that the set of extremal non-commutative bipartite graphs has a one-to-one 
correspondence to the set of extremal quantum channels has greatly simplified the 
simulation problem. How to use this property to simplify the assisted-communication 
problem is still unclear.

\subsection{Example: Non-commutative bipartite cq-graphs with two output states}
\label{2-state-example}
Here we will examine our above findings of one-shot and asymptotic
capacities and simulation costs for the simplest possible
cq-channel, which has only two inputs and two pure output states $P_i = \proj{\psi_i}$,
w.l.o.g.
\begin{align*}
  \ket{\psi_0} &= \alpha\ket{0} + \beta\ket{1}, \\
  \ket{\psi_1} &= \alpha\ket{0} - \beta\ket{1},
\end{align*}
with $\alpha \geq \beta = \sqrt{1-\alpha^2}$. In fact, we shall assume
$\alpha > \beta > 0$ since the two equality cases are trivial
(noiseless classical channel and completely noisy channel, respectively).
Note $|\bra{\psi_0}\psi_1\rangle| = \alpha^2-\beta^2 = 2\alpha^2-1$;  
the non-commutative bipartite cq-graph 
$K={\rm span}\{\ket{\psi_0}\!\bra{0},\ket{\psi_1}\!\bra{1}\}$.
We can work out all the optimization problems introduced before:
\begin{align}
  \label{twostate-cq-1-shot}
  \U(K)          &=    1, \\
  \label{twostate-cq-2-shot}
  \U(K\ox K)     &\geq \frac{1}{\alpha^4+\beta^4},~\mbox{if}~|\braket{\psi_0}{\psi_1}|\leq \frac{1}{\sqrt{2}},                                 \\
  \label{twostate-cq-n-shot}
  \U(K^{\ox n})  &\geq \frac{1}{\alpha^{2n}+\beta^{2n}}
                  \geq \frac{1}{2\alpha^{2n}}\ \text{ for sufficiently large } n,    \\
  \label{twostate-cq-aram}
  \Aram(K)       &=    \frac{1}{\alpha^2} = \frac{2}{1+|\bra{\psi_0}\psi_1\rangle|}, \\
  \label{twostate-cq-CminE}
  C_{\min {\rm E}}(K)  &=    H(\alpha^2,\beta^2),                                    \\
  \label{twostate-cq-cost}
  S_{0,\NS}(K)   &=\S(K)= 1+\frac12 \|P_0-P_1\|_1 = 1+2\alpha\beta.
\end{align}

Eq. (\ref{twostate-cq-n-shot}) directly gives us $
C_{0,\NS}(K)=\log {1}/{\alpha^2}.$
 Since the signal ensemble is symmetric under the Pauli $Z$ unitary 
(which exchanges two output states), it is easy to evaluate $\Aram(K) = 1/\a^2$,
Eq. (\ref{twostate-cq-aram}). Also $C_{\min {\rm E}}(K)$ is easy to compute, yielding Eq. (\ref{twostate-cq-CminE}). For pure state cq-channels $\cN$, by dephasing the input of any channel with Kraus operators in $K = \operatorname{span}\{ \ket{\psi_i}\!\bra{i}: i=0,1 \}$,
we get a simulation of $\cN$ itself, hence
\[
  \S(K) = 1+\frac12\|\psi_0-\psi_1\|_1 = 1+2\alpha\beta,
\]
proving Eq. (\ref{twostate-cq-cost}). Noticing that $\cN$ is the unique extremal channel in $K$, we can also apply Theorem \ref{extremal-channel-cost} to obtain Eqs. (\ref{twostate-cq-CminE}) and (\ref{twostate-cq-cost}) directly.

So we have
\[\begin{split}
  C_{0,\NS}(K)  = \log \Aram(K)
               &= \log \frac{2}{1+|\bra{\psi_0}\psi_1\rangle|}                    \\
               &\leq H\left( \frac{1+|\bra{\psi_0}\psi_1\rangle|}{2}, 
                          \frac{1-|\bra{\psi_0}\psi_1\rangle|}{2} \right) = C_{\min {\rm E}}(K) \\
               &\leq \log \left(1+\frac12 \|\psi_0-\psi_1\|_1\right) = S_{0,\NS}(K),
\end{split}\]
and both inequalities become strict when $0<|\braket{\psi_0}{\psi_1}|<1$ (refer to Fig. 4).

\begin{figure}[ht]
  \includegraphics[width=11cm]{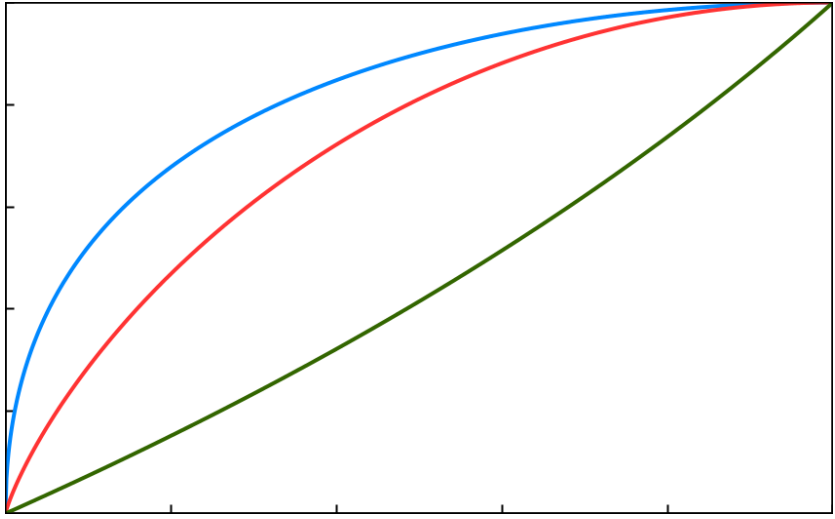}
  \caption{Comparison between $C_{0,\NS}$ (green), $C_{\min {\rm E}}$ (red) and $S_{0,\NS}$ (blue)
    for the cq-channel of two pure states, as a function of $0\leq \beta^2 \leq \frac12$.}
\end{figure}

The largest effort goes into calculating or bounding the numbers
$\U(K^{\ox n})$ in  Eq. (\ref{twostate-cq-n-shot}). In this case we can obtain much better lower bounds 
(only at most one bit less than optimal value by $n$ uses) compared to
Theorem~\ref{thm:cq-capacity}.

\medskip\noindent
\emph{One copy $n=1$}: Let $Q_i = \1-P_i$ be the projection orthogonal to $P_i$ 
(which is rank-one), so that $R_i = r_i Q_i$ with a number $r_i \leq s_i$. 
Because of the $Z$-symmetry,
\[ Z P_0 Z=P_1, Z Q_0 Z=Q_1. \]
We can symmetrize any solution and assume $s_0=s_1=s$ and $r_0=r_1=r$. 
Then the normalization condition reads
\begin{align*}
  \1 &= s(P_0+P_1) + r(Q_0+Q_1) \\ 
     &= s(\a^2\proj{0}+\b^2\proj{1}) + r(\b^2\proj{0}+\a^2\proj{1}),
\end{align*}
which implies $r=s=1/2$. Hence the maximum value of $\U(K)=1$.

\medskip\noindent

\emph{Many copies $n>1$}: In this case it is difficult to find the optimum, but it is enough that we show the achievability of $1/(\a^{2n}+\b^{2n})$ for sufficiently large $n$. Note that this already implies that the SDP for the zero-error number of messages is not multiplicative! Somehow, what's happening is that the normalization condition of $\sum_i (s_i P_i + R_i) = \1$ is a non-trivial constraint because $R_i$ has to be supported on the orthogonal complement of $P_i$; this hurts us in the case $n=1$. Now in the case of many
copies, $P_i$ is a tensor product of single-system projections, hence the orthogonal complement is asymptotically dominating. We have seen how these considerations help in understanding the general case.

We have $2^n$ states $P_{i^n} = P_{i_1} \ox P_{i_2} \ox \cdots \ox P_{i_n}$, indexed by $n$-bit strings $i^n$, which are related by qubit-wise $Z$-symmetry:

 $$P_{i^n} = Z^{i^n} P_{0^n} Z^{i^n},$$
which motivates that we find $R_{0^n}$ (orthogonal to $P_{0^n}$) and define

 $$ R_{i^n} := Z^{i^n} R_{0^n} Z^{i^n}.$$

Let $$P^{(n)}_w=\sum_{wt(i^n)=w}\ketbra{i^n}{i^n},~\ket{\Phi_w^{(n)}}={n\choose w}^{-1/2}\sum_{wt(i^n)=w}\ket{i^n}, ~0\leq w\leq n,$$
where $wt(i^n)$ is the weight (the number of  $1$'s) of the $n$-bit string $i^n$. We also denote
$$Q_w^{(n)}=P^{(n)}_w-\ketbra{\Phi_{w}^{(n)}}{\Phi_{w}^{(n)}},$$ 
which is always a projection for $1\leq w\leq n-1$ and vanishing for $w=0, n$.

For $\ket{\psi_0}=\alpha \ket{0}+\beta \ket{1}$, we have 
$$\ket{\psi_0}^{\ox n}=\sum_{w=0}^n x_w \ket{\Phi_w^{(n)}},~x_w={n\choose w}^{1/2}\alpha^{n-w}\beta^{w}.$$

Clearly
$$Q_w^{(n)}\ket{\psi_0}^{\ox n}=0~\mbox{and}~Q_w^{(n)}\ket{\psi_0^\perp}^{\ox n}=0,~1\leq w\leq n-1$$
for $\ket{\psi_0^\perp}=\beta\ket{0}-\alpha\ket{1}$.

Now we set all $s_{i^n}$ equal to $s$, and propose the following ansatz 
$$R_{0^n}=s(\ketbra{\psi_0^\perp}{\psi_0^\perp})^{\ox n}+\sum_{w=1}^{n-1}c_w Q_w^{(n)},$$
where $s$ and $c_w$ are non-negative eigenvalues to be determined. Note that $R_{0^n}$ is automatically supported on the orthogonal complement of $P_{0^n}$.

We need 
$$\sum_{i^n} (s P_{i^n} + R_{i^n})=\1_{2^n},$$
or equivalently,

$$\sum_{i^n} Z^{i^n}(s(\ketbra{\psi_0}{\psi_0})^{\ox n}+s(\ketbra{\psi_0^\perp}{\psi_0^\perp})^{\ox n}+\sum_{w=1}^{n-1}c_w Q_w^{(n)})Z^{i^n}=\1_{2^n},~~~(*)$$

Using the $Z^n$ symmetry, we can calculate
$$\sum_{i^n}Z^{i^n}(\ketbra{\psi_0}{\psi_0})^{\ox n}Z^{i^n}=2^n\sum_{w=0}^n \alpha^{2n-2w}\beta^{2w}P_w^{(n)}.$$ Similarly
$$\sum_{i^n}Z^{i^n}(\ketbra{\psi_0^\perp}{\psi_0^\perp})^{\ox n}Z^{i^n}=2^n\sum_{w=0}^n \alpha^{2w}\beta^{2n-2w}P_w^{(n)},$$
and
$$\sum_{i^n}Z^{i^n}Q_{w}^{(n)} Z^{i^n}=2^n(1-{n \choose w}^{-1})P_w^{(n)}.$$

Applying the property that $\sum_{w=0}^n P_w^{(n)}=\1_{2^n}$, we can see that (*) is equivalent to

$$2^ns(\alpha^{2n}+\beta^{2n})=1,$$ and 
$$2^n c_w(1-{n \choose w}^{-1})+2^ns(\alpha^{2w}\beta^{2n-2w}+\alpha^{2n-2w}\beta^{2w})=1,~1\leq w\leq n-1.$$
Then
$$s=2^{-n}(\alpha^{2n}+\beta^{2n})^{-1},~\mbox{and}~ c_w=(1-{n \choose w}^{-1})^{-1}s(\alpha^{2n}+\beta^{2n}-\alpha^{2w}\beta^{2n-2w}-\alpha^{2n-2w}\beta^{2w})), 1\leq w\leq n-1.$$
We need $R_{0^{n}}\leq s (\1_{2^n}-P_{0^n})$, which is equivalent to $c_w\leq s$, or 
$$(1-{n \choose w}^{-1})^{-1}\cdot (\alpha^{2n}+\beta^{2n}-\alpha^{2w}\beta^{2n-2w}-\alpha^{2n-2w}\beta^{2w}))\leq 1.$$
The first term of the left hand side (LHS) of the above inequality achieves the maximum when $w=1$, and the second term of the LHS reaches the maximum when $w=\frac{n}{2}$. So we only need 
$$(1-{n\choose 1}^{-1})^{-1}(\alpha^n-\beta^n)^2\leq 1,$$
or 
$$\alpha^n-\beta^n\leq \sqrt{\frac{n-1}{n}},$$
which is always satisfied when $n$ is large enough.
Note that when $n=2$, we have 
$$\alpha^2-\beta^2=|\braket{\psi_0}{\psi_1}|\leq \sqrt{\frac{1}{2}}.$$ The above constraint is tight in the sense 
$\U(K\ox K)\geq  {1}/(\alpha^4+\beta^4)$ if $|\braket{\psi_0}{\psi_1}|\leq \sqrt{\frac{1}{2}}$; otherwise $\U(K\ox K)=1$. More precisely, we have $\U(K\ox K)=\frac{4}{3}$ when $|\braket{\psi_0}{\psi_1}|\leq \sqrt{\frac{1}{2}}$, and $\U(K\ox K)>{1}/{(\alpha^4+\beta^4)}$ when $|\braket{\psi_0}{\psi_1}|<\sqrt{\frac{1}{2}}$.
\qed

\medskip
Let us consider now the more general case with two output states 
$\rho_0$ and $\rho_1$ with projections $P_0$ and $P_1$, respectively.  
Denote 
$$F_{\max}(\rho_0,\rho_1)= F_{\max}(P_0,P_1)
                         := \max \{|\bra{\psi_0}\psi_1\rangle|: \psi_i\in K_i, i=0,1 \}.$$
$F_{\max}(\rho_0,\rho_1)$ is known as the maximal fidelity between 
$\rho_0$ and $\rho_1$, but depends only on their supports $K_0$ and $K_1$.
A key property of the maximal fidelity is the following \cite{DFY2009}:

\begin{proposition}
  There exists a CPTP map $\cT$ such that $\cT(\rho_0)=\psi_0$ and 
  $\cT(\rho_1)=\psi_1$ if and only if 
  $F_{\max}(\rho_0,\rho_1)\leq |\bra{\psi_0}\psi_1\rangle|$.
  \qed
\end{proposition}

Applying this result, we can show that the case of two general output states 
$\rho_0$ and $\rho_1$ is simply equivalent to the case of two pure output 
states $\psi_0$ and $\psi_1$ such that 
$|\bra{\psi_0}\psi_1\rangle|=F_{\max}(\rho_0,\rho_1)$.  
Thus we have
\begin{align*}
  C_{0,\NS}(K)        &=\log\Aram(K)
                       = \log \frac{2}{1+F_{\max}},                    \\
  C_{\min {\rm E}}(K) &= H\left(\frac{1+F_{\max}}{2},\frac{1-F_{\max}}{2}\right), \\
  S_{0,\NS}(K)        &= \log \S(K) 
                       = 1+\sqrt{1-F_{\max}^2}.
\end{align*}

\section{An operational interpretation of the Lov\'{a}sz number}
\label{sec:theta}
As we have seen, different non-commutative bipartite graphs $K$, even 
cq-graphs, having the same confusability graph $G$, can
have different assisted zero-error capacity $C_{0,\NS}(K)$ and simulation
cost $S_{0,\NS}(K)$; c.f.~the last subsection \ref{2-state-example} in the
previous section.

A classical undirected graph is given by $G=(V,E)$, where $V=\{1,...,n\}$ 
is the set of vertices, and $E\subset V\times V$ is the set of edges. As 
shown in previous work \cite{DSW2010}, $G$ is naturally associated with 
a non-commutative graph, denoted $S_G$, via the following way:
\begin{equation}\label{sg-graph}
  S_G= {\rm span}\{\ketbra{i}{j}: i\sim j\},
\end{equation}
where $i\sim j$ means confusability: $i=j$ or
$\{i,j\} \in E$ is an edge of the graph \cite{Lovasz1979}.

Hence the questions we are facing are the maximum and minimum
$C_{0,\NS}(K)$ and $S_{0,\NS}(K)$ over all cq-graphs $K$ with
$K^\dag K<S_G$. While the maximum is clearly $\log|B|$ for
both quantities, the minima turn out to be much more interesting.
We restate here the main result we will go on to prove in this section.

\medskip\noindent
{\bf Theorem~\ref{thm:theta}}\ 
\textit{For any classical graph $G$, the Lov\'asz number $\vartheta(G)$
  is the minimum zero-error classical capacity assisted 
  by quantum no-signalling correlations of any cq-channels that have 
  $G$ as non-commutative graph, \emph{i.e.}
  \[
    \log\vartheta(G)=\min\bigl\{ C_{0,{\rm NS}}(K): K^\dag K<S_G \bigr\},
  \]
  where the minimization is over cq-graphs $K$.
  \\
  \indent In particular, equality holds for any cq-channel $i\rightarrow \proj{\psi_i}$ 
  such that $\{\ket{\psi_i}\}$ is an optimal orthogonal representation for $G$
  in the sense of Lov\'{a}sz' original definition \cite{Lovasz1979}.}

\medskip
The proof of this result is achieved by combining two facts about the 
semidefinite packing number for cq-channels: 1) $\Aram(K)$ gives 
the zero-error classical capacity assisted with no-signaling correlations 
for a cq-channel; 2) the Lov\'asz number $\vartheta(G)$ of a graph 
is given by the minimization of $\Aram(K)$ for all non-commutative bipartite 
graph $K$ that generate the same confusability graph $G$. The first 
fact has been proven in Theorem \ref{c-q-capacity-aram}, so
we focus on the second for the rest of the section.

Let $K=\operatorname{span}\{E_i:1\leq i\leq n\}$ be a Kraus operator space with 
$\tr(E_i^\dagger E_j)=\d_{ij}$, and let 
$\ket{\Phi}=\sum_{i=1}^d\ket{i}\ket{i}$ be the non-normalized maximally entangled 
state. Then $P$, the projection on the support of the Choi-Jamio\l{}kowski state,
can be written as
\begin{equation}
  \label{Kraus-Projection}
  P_{AB}=\sum_{i=1}^n (\1\ox E_i)\proj{\Phi}(\1\ox E_i)^\dagger.
\end{equation}

We can rewrite the semidefinite packing number $\Aram(K)$ using these Kraus 
operators $\{E_i\}$:
\begin{equation}
  \label{Primal}
  \Aram(K)=\max  \tr R \ \text{s.t. }\ \sum_{i=1}^n E_i R E_i^\dagger\leq \1_B,\ R\geq 0.
\end{equation}
Note that $\{E_i\}$ spans a valid Kraus operator space of a quantum channel. 
So $\sum_i E_i^\dagger E_i > 0$ (positive definite). The dual 
SDP is
\begin{equation}
  \label{Dual}
  \Aram(K)=\min \tr T \ \text{ s.t. }\ \sum_{i=1}^n E_i^\dag T E_i\geq \1_A,\ T\geq 0,
\end{equation}
and we can easily verify that both the primal and the dual are strictly feasible 
by choosing $R=0$ and $T=\lambda \1_B$ (here $\lambda>0$ is sufficiently large), 
respectively. Hence strong duality holds. 

We will start by deriving some minimax representations of $\Aram(K)$. 
Let us introduce
\begin{align*}
\widehat{\Aram}(K) := \min_{\rho} \ll_{\max}(\cN(\rho))
     &= \min_{\rho} \max_{\sigma}\tr (\cN(\rho)\sigma)         \\
     &= \max_{\sigma} \min_{\rho}\tr (\cN^\dagger(\sigma)\rho) \\
     &= \max_{\sigma}\lambda_{\min}\bigl(\cN^\dagger(\sigma)\bigr).
\end{align*}
where $\lambda_{\max}(T)$ and $\lambda_{\min}(T)$ represent the maximal and 
the minimal eigenvalues of a Hermitian operator $T$, respectively,
and $\cN(\rho)=\sum_i E_i\rho E_i^\dagger$, 
$\cN^\dagger(\sigma)=\sum_i E_i^\dagger \sigma E_i$ are CP maps
(but not necessarily trace or unit preserving);
$\rho$ and $\sigma$ range over all density operators, 
i.e.~$\rho,\sigma\geq 0$ and $\tr \rho = \tr \sigma = 1$.

In the second and the fourth equalities above, we have employed the following well-known
characterizations:
\[
  \lambda_{\max}(T)=\max_{\rho}\tr \rho T\ \text{and}\ 
  \lambda_{\min}(T)=\min_{\sigma}\tr \sigma T.
\]

In the third equality we have employed the obvious equality 
$\tr (\cN(\rho)\sigma)=\tr (\rho N^\dagger(\sigma))$, and von Neumann's 
minimax theorem \cite{Sion:minimax}, 
since $\tr \cN(\rho)\sigma$ is a linear function with respect to $\rho$ and $\sigma$, 
and $\rho$ and $\sigma$ range over convex compact sets.

\begin{lemma}
  \label{lemma:aram-dual}
  Under the above definitions,
  \[
    \Aram(K) = \frac{1}{\widehat{\Aram}(K)},
  \]
  for any non-commutative bipartite graph $K$.
\end{lemma}
\begin{proof}
Suppose that $\Aram(K)=\tr R_0$ for some $R_0\geq 0$. 
Let us construct a density operator $\rho_0=\frac{R_0}{\tr R_0}$.  
By the assumption $\sum_i E_i R_0 E_i^\dagger\leq \1_B$, we have
\[
  \sum_i E_i \rho_0E_i^\dagger\leq \frac{1}{\tr R_0} \1_B,
\]
or equivalently
\[
  \cN(\rho_0)\leq \frac{1}{\Aram(K)} \1_B.
\]
By the definition of $\widehat{\Aram}(K)$, we have 
\[
  \widehat{\Aram}(K) \leq \ll_{\max}(\cN(\rho_0))
                     \leq \frac{1}{\Aram(K)}.
\]
Conversely, suppose that $\widehat{\Aram}(K)=\ll_{\max}(\cN(\rho_0))$ for 
some density operator $\rho_0$. Then we have $\cN(\rho_0)\leq \widehat{\Aram}(K)\1_B$, 
and thus
$$\cN\bigl( \rho_0/\widehat{\Aram}(K) \bigr) \leq \1_B.$$
That means $\rho_0/\widehat{\Aram}(K)$ is a feasible solution of the SDP defining
the semidefinite packing number. 
By the definition of $\Aram(K)$, we know that 
$$\Aram(K)\geq \tr \rho_0/\widehat{\Aram}(K) = \frac{1}{\widehat{\Aram}(K)},$$
concluding the proof.
\end{proof}

\medskip
Focussing on the special class of cq-channels, which are in some sense the 
direct quantum generalizations of classical channels, we are now ready for

\medskip
\begin{beweis}{of Theorem~\ref{thm:theta}}
Let $K$ correspond to a cq-channel $i\mapsto  \rho_i$, and  let $P_i$ be the projection on 
the support of quantum states $\rho_i$.  It will be more convenient to study $\widehat{\Aram}(K)$ instead of $\Aram(K)$.  Actually, applying the above minimax representation to this special case, we have
$$\widehat{\Aram}(K) = \min_{\{t_i\}} \max_{\sigma} \tr \sigma\left(\sum_i t_i P_i\right)
                     = \max_{\sigma} \min_{\{t_i\}} \tr \sigma\left(\sum_i t_i P_i\right)
                     = \max_{\sigma} \min_{i} \tr \sigma P_i,$$
where $\{t_i\}$ ranges over probability distributions, and $\sigma$ 
ranges over density operators. The right-most expression  
motivates us to introduce some notations. 

By definition, $\{P_i\}$ is an orthogonal representation 
(OR) of the confusability graph $G$ induced by the cq-channel. The value of an OR $\{P_i\}$ 
is defined as follows: 
$$\eta(\{P_i\})=\max_{\sigma}\min_i \tr P_i\sigma.$$
We introduce the following function of a graph $G$,
$$\eta(G)=\max_{\{P_i\}}\eta(\{P_i\}),$$
where the maximization ranges over all possible ORs of $G$. Clearly, 
if we require that an OR consists of only rank-one projections and 
$\sigma$ takes only rank-one projection, then $\eta(G)=\vartheta(G)^{-1}$, 
the reciprocal of the Lov\'asz number of $G$ \cite{Lovasz1979}. 
It has been shown in \cite{CSW2010} that even allowing $P_i$ to be general 
projection but $\sigma$ to be rank-one projection, there is no 
difference between $\eta(G)$ and $\vartheta(G)^{-1}$. However, 
if $\sigma$ is a mixed state, we can only have 
$\eta(G)\geq \vartheta(G)^{-1}$. Interestingly, we can show that 
equality does hold. In fact, it is evident that if $\{P_i\}$ is an 
OR for a graph $G$, then $\{P_i\ox \1_B\}$ remains an OR for the 
same graph, where $B$ is any auxiliary system.
Now the value of $\{P_i\}$ with respect to general mixed states, is 
the same as the value of $\{P_i\ox \1_B\}$ with respect to pure states. That is,
$$\max_{\sigma} \min_{i} \tr P_i\sigma = \max_{\Psi} \min_i  \tr ((P_i\ox \1_B) \ket{\Psi_{AB}}\bra{\Psi_{AB}}),$$ 
where $\sigma=\tr_B  \ket{\Psi_{AB}}\bra{\Psi_{AB}}$.  The above equality follows directly 
from the fact that
$$\tr P_i\sigma = \tr ((P_i\ox \1_B) \ket{\Psi_{AB}}\bra{\Psi_{AB}}),$$ 
where $\ket{\Psi_{AB}}$ is any purification of $\sigma$.

Summarizing, we have
\[
  \min_K \Aram(K) = \min_K \frac{1}{\widehat{\Aram}(K)} = \frac{1}{\eta(G)} = \vartheta(G),
\]
and we are done.
\end{beweis}

\medskip
As a final comment on the above proof, 
note that for a fixed OR $\{P_i\}$, we cannot always choose the optimal handle $\sigma$ as a pure state-- the restriction to rank-one projectors
and pure state handle only emerges as we optimize over both elements.

\medskip
We would like to interpret Theorem~\ref{thm:theta} to say
intuitively that the zero-error capacity of a graph $G$ assisted by
no-signalling correlations is $\vartheta(G)$. The problematic 
part of such a manner of speaking is that there are many 
cq-graphs $K$ with the same confusability graph $G$, but the
no-signalling assisted capacity may vary with these $K$. 

However, note that for any (finite) family of cq-graphs $K^{(\alpha)}$
such that ${K^{(\alpha)}}^\dagger K^{(\alpha)} = S_G$, 
$K^{(\alpha)} = \sum_i \ket{i}\ox K_i^{(\alpha)}$,
we can construct a cq-graph $K$ that ``dominates'' all of the $K^{(\alpha)}$
in the sense that any no-signalling assisted code for $K$ can be
used directly for $K^{(\alpha)}$ because actually $K^{(\alpha)} < K$:
\[
  K = \sum_i \ket{i}\ox K_i,\quad K_i = \bigoplus_\alpha K_i^{(\alpha)}.
\]
By going from direct sums to direct integrals, we can thus construct a
universal cq-graph
\[
  \widetilde{K}(G) = \sum_i \ket{i}\ox \widetilde{K}_i
\]
which dominates \emph{all} $K^{(\alpha)}$ with confusability graph $G$, 
in fact contains them up to isomorphism. Any no-signalling assisted
code for this object will deal in particular with every eligible channel
simultaneously. The only caveat is that the $\widetilde{K}_i$ are subspaces 
in an a priori infinite dimensional Hilbert space, and all of our
proofs (in particular that of Theorem~\ref{thm:cq-capacity})
require finite dimension as a technical condition. We conjecture
however that the capacity result of Theorem~\ref{thm:cq-capacity} still 
holds in that setting.

\bigskip
{\bf Minimum simulation cost of a confusability graph.}
Just as we were looking at the smallest zero-error capacity over all
cq-channels with a given confusability graph $G$ in this section, we 
can study the minimum simulation cost over all cq-graphs
$K$ with $K^\dagger K <S_G$. 
To be precise, we are interested in 
\[
  \Sigma(G) := \inf\{ \Sigma(K) : K \text{ cq-graph with } K^\dagger K< S_G \},
\]
and the asymptotic simulation cost (regularization)
\[
  S_{0,\NS}(G) := \lim_{n\rightarrow\infty} \frac1n \log \Sigma(G^n),
\]
where $G^n = G\times\cdots\times G$ denotes the $n$-fold strong graph product.
The latter limit exists and equals the infimum because evidently
$\log \Sigma(G\times H) \leq \log \Sigma(G) + \log \Sigma(H)$.

For a cq-graph $K$, we have 
\begin{equation}
  \label{eq:K-chain}
  \log \Aram(K) \leq C_{\min E}(K) \leq \log \Sigma(K),
\end{equation}
in fact for every eligible cq-channel $\cN$ with non-commutative 
bipartite graph $K$,
\[
  \log \Aram(K) \leq C(\cN) = C_E(\cN) \leq \log \Sigma(\cN).
\]
The reason is that $\log \Aram(K) = C_{0,\NS}(K)$ is the zero-error
capacity assisted by no-signalling correlations (Theorem \ref{thm:cq-capacity});
while $C_E(\cN)=C(\cN)$ is the Holevo (small-error) capacity,
which is the same as the entanglement-assisted capacity for
cq-channels and which is not increased by any other available
no-signalling resources because of the Quantum Reverse Shannon 
Theorem \cite{BDHS+2009,QRST-simple}; 
and $\log \Sigma(\cN) = -H_{\min}(X|B)_J$ is the perfect 
simulation cost of the channel when assisted by no-signalling resources,
with the Choi-Jamio\l{}kowski matrix $J = \sum_i \proj{i} \otimes P_i$,
Eq.~(\ref{eq:G0-equals-sigma}).

In Eq.~(\ref{eq:K-chain}), $C_{\rm \min E}(K) = C_{\min}(K)$ because
$K$ is a cq-graph, so we only need to consider cq-channels 
$\cN : i \mapsto \rho_i$ 
in the minimization, for which
\[
  C_{\rm E}(\cN) = C(\cN) = \max_{P_X} I(X:B) 
                    = \max_{P_X} S\left(\sum_i P_X(i)\rho_i \right) - \sum_i P_X(i) S(\rho_i). 
\]
Letting now
\[
  C_{\min}(G) := \inf \{ C_{\min}(K) : K \text{ cq-graph with } K^\dagger K<S_G \},
\]
we then have the following additivity result.

\begin{lemma}
  \label{lemma:C-min-add:graph}
  For any two graphs $G$ and $H$,
  \[
    C_{\min}(G\times H) = C_{\min}(G) + C_{\min}(H).
  \]
\end{lemma}
\begin{proof}
The subadditivity, $C_{\min}(G\times H) \leq C_{\min}(G) + C_{\min}(H)$,
is evident from the definition, because if $K^\dagger K<S_G$
and $L^\dagger L < S_H$, 
then $(K\otimes L)^\dagger (K\otimes L)< S_{G \times H}$.

It remains to show the opposite inequality ``$\geq$''.
This relies crucially on the minimax identity
\begin{equation}
  \label{eq:C_min:minimax}
  C_{\min}(G) = \inf_{\cN} \max_{P_X} I(X:B)
              = \max_{P_X} \inf_{\cN} I(X:B),
\end{equation}
where the maximum is over probability distributions $P_X$
and the infimum is over cq-channels $\cN$ with confusability
graph contained in $G$. This is a special case of Sion's 
minimax theorem~\cite{Sion:minimax}, since the Holevo mutual
information is well-known to be concave in $P_X$ and convex
in $\cN$, while the domain of $P_X$ is the convex compact
simplex of finite probability distributions and the domain
of $\cN$ is an infinite-dimensional convex set.

Now, for a cq-channel $\cN:ij \mapsto \rho_{ij}$
with confusability graph contained in $G\times H$, and an
arbitrary distribution $P_{XY}$ of the two input variables $X$
and $Y$, we have
\begin{equation}
  \label{eq:chain-rule}
  I(XY:B) = I(X:B) + I(Y:B|X)
          = I(X:B) + \sum_i P_X(i) I(Y:B|X=i).
\end{equation}
Here, the first term refers to the cq-channel
\[
  \overline{\cN} : i \mapsto \overline{\rho}_i = \sum_j P_{Y|X}(j|i) \rho_{ij},
\]
while the $i$-th summand in the second term sum refers to the cq-channel
\[
  \cN_i : j \mapsto \rho_{ij}.
\]
Note that $\overline{\cN}$ is eligible for $G$ (since $i \not\sim i'$ implies
$ij \not\sim i'j'$ for all $j,j'$, hence $\rho_{ij} \perp \rho_{i'j'}$), 
while similarly for all $i$, $\cN_i$ is eligible for $H$.
Thus, in Eq.~(\ref{eq:chain-rule}), we can take the infimum over
eligible cq-channels, to obtain
\[
  \inf_{\cN} I(XY:B) \geq \inf_{\cM_1} I(X:C_1) + \inf_{\cM_2} I(Y:C_2)
                     =    \inf_{\cM_1,\cM_2} I(X'Y':C_1 C_2),
\]
where the minimizations are over cq-channels $\cN$ eligible for $G\times H$,
$\cM_1$ eligible for $G$ and $\cM_2$ eligible for $H$, whereas
$X'$ and $Y'$ are independent copies of $X$ and $Y$, i.e.~they are
jointly distributed according to $P_X \times P_Y$.
Now, taking the maximum over distributions $P_{XY}$ completes
the proof because of the minimax formula (\ref{eq:C_min:minimax}).
\end{proof}

As a corollary, we get the following chain of inequalities:
\begin{equation}
  \label{eq:min-over-graphs}
  \log \vartheta(G) \leq C_{\min}(G) \leq S_{0,\NS}(G) \leq \log \Sigma(G) \leq \log \alpha^*(G).
\end{equation}
Note that it may be true that $S_{0,\NS}(G) = \log \Sigma(G)$, 
but to prove this we would need to show the additivity relation
$\log \Sigma(G\times H) = \log \Sigma(G) + \log \Sigma(H)$, which 
remains unknown.
Another observation is that if in the respective minimizations,
the channels are restricted to classical channels, then the results 
of \cite{CLMW2011} show that 
\[
  \min_K \log \Aram(K) = \min_{\cN} C(\cN) = \min_{\cN} \Sigma(\cN) = \log \alpha^*(G).
\]

\medskip
We now demonstrate by example that the rightmost inequality in (\ref{eq:min-over-graphs})
can be strict, when quantum channels are considered. 
Namely, for even $n$ let $G = \overline{H_n}$ the complement of the Hadamard 
graph $H_n$, whose vertices are the vectors $\{\pm 1\}^n$ and two vectors $v,w$ 
are adjacent in $H_n$ if and only if they are orthogonal in the Euclidean sense, 
$v^\top w = 0$.
In other words, the cq-channel $\cN:v \mapsto \frac{1}{n}\proj{v}$
has confusability graph $G$. Since the output dimension is $n$, we see
from this that $\Sigma(G) \leq \Sigma(\cN) \leq n$, which happens to
coincide with the Lov\'{a}sz number, $\vartheta(G) = n$, hence
$\log \vartheta(G) = C_{\min}(G) = S_{0,\NS}(G) = \log \Sigma(G) = \log n$.
On the other hand, the clique number of $G$ is known to be upper bounded
$\omega(G) \leq 1.99^n$ \cite{FranklRoedl}, hence
$\alpha^*(G) \geq \frac{|G|}{\omega(G)} \geq 1.005^n$, meaning
$\log \alpha^*(G) \geq \Omega(n)$.

We think that also the leftmost inequality in (\ref{eq:min-over-graphs})
can be strict. But although the pentagon $G=C_5$ seems to be a candidate,
for which we conjecture (based on some ad hoc calculations) that
$C_{\min}(G) = S_{0,\NS}(G) = \log \Sigma(G) = \log \alpha^*(G) = \log\frac52$,
whereas $\log \vartheta(G) = \frac12 \log 5$, a rigorous proof of this 
has so far eluded us. 

Whether the other two inequalities can be strict remains an open question.

\section{Feasibility of zero-error communication via general non-commutative bipartite graph assisted by quantum no-signalling correlations}
\label{sec:feasibility}
Given a non-commutative bipartite graph $K$, it is important to know when $K$ 
is able to send classical information exactly in the presence of quantum 
no-signalling correlations. It turns out that these channels can be precisely
characterized; we will start with cq-graphs.

\begin{theorem}
\label{feasibility-cq}
Let $K$ be a non-commutative bipartite cq-graph specified by a set of 
projections $\{P_i:1\leq i\leq n\}$ with supports $K_i = \operatorname{supp} P_i$. 
Then the following are equivalent:
\renewcommand{\theenumi}{\roman{enumi}}
\begin{enumerate}
  \item $C_{0,{\rm NS}}(K)>0$;
  \item $\Aram(K)>1$;
  \item $\bigcap_{i} K_i = \{0\}$.
\end{enumerate}
\renewcommand{\theenumi}{\arabic{enumi}}
\end{theorem}
\begin{proof}
The equivalence of i) and ii) follows directly from Proposition \ref{c-q-capacity-aram}. 

The equivalence of ii) and iii) is only a simple application of the SDP of $\Aram(K)$ 
for bipartite cq-graphs. First, we show that iii) implies ii). By contradiction, 
assume that the intersection of supports of $P_i$ is empty while  $\Aram(K)=1$. 
Recall that the dual SDP of $\Aram(K)$ is given by
$$\Aram(K)=\min \tr T\ \text{ s.t. }\ \tr P_i T\geq 1,\ T\geq 0, 1\leq i\leq n.$$
Then $\Aram(K)=1$ implies that we can find $T_0\geq 0$ such that $\tr T_0=1$ 
and $\tr P_i T_0\geq 1$ for any $i$. Clearly $T_0$ is a density operator, and we 
should have $\tr P_i T_0=1$ for any $i$. The only possibility is that $T_0$ is 
in the intersection of the supports of $P_i$, which is a contradiction.
Now we turn to show that ii) implies iii). Again by contradiction, assume that 
$\Aram(K)>1$ while the intersection of supports is nonempty. Then we can find a 
pure state $\ket{\psi}$ from the intersection such that $P_i\geq \proj{\psi}$. So
$$\1\geq \sum_i s_i P_i\geq \left(\sum_i s_i\right) \proj{\psi},$$
thus any feasible solution should have $\sum_i s_i\leq 1$, which indeed implies 
that $\Aram(K)=1$.
\end{proof}

\begin{theorem}
\label{feasibility-K}
Let $K$ be a non-commutative bipartite graph with Choi-Jamio\l{}kowski 
projection $P_{AB}$, and let $Q_{AB}=\1_{AB}-P_{AB}$ be the orthogonal 
complement of $P_{AB}$. Then the following are equivalent:
\renewcommand{\theenumi}{\roman{enumi}}
\begin{enumerate}
  \item $C_{0,{\rm NS}}(K)>0$;
  \item $\Aram(K)>1$;
  \item $\tr_A P_{AB}<d_A\1_B$;
  \item $\tr_A Q_{AB}$ is positive definite.
\end{enumerate}
\renewcommand{\theenumi}{\arabic{enumi}}
As a matter of fact, we have 
\[
  C_{0,{\rm NS}}(K) \geq \log \frac{d_A}{\|\tr_A P_{AB}\|_\infty}
  \quad\text{and}\quad
  \Aram(K)\geq  \frac{d_A}{\|\tr_A P_{AB}\|_\infty}.
\]
\end{theorem}

\begin{proof} 
The meaning of i) and ii) are very clear, while iii) and iv) need some explanation. 

Essentially, iv) means we can find a  CP map from $B$ to $A$ with 
Choi-Jamio\l{}kowski matrix $V_{AB}$  supporting on some subset of 
$Q_{AB}$. Note that in the one-shot SDP formulation of $\U(K)$, we 
need $V_{AB}$ to be a CPTP map. However, the trace-preserving condition 
is not necessary for asymptotic case but only $\tr_A V_{AB}$ is 
positive definite, which is the most nontrivial part of this theorem. 

iii) is directly equivalent to iv) as we have $P_{AB}+Q_{AB}=\1_{AB}$, 
hence 
$$\tr_A P_{AB}+\tr_A Q_{AB}=d_A\1_B.$$ 

In the following we only focus on i), ii), and iii).
The equivalence of ii) and iii) is straightforward, simply noticing that $\1_A/\|P_B\|_\infty$ 
is a feasible solution to the primal SDP for $\Aram(K)$, where $P_B=\tr_A P_{AB}$.

The equivalence of i) and iii) is much more difficult and non-trivial. In the following we want to explain a little bit more about this equivalence as there are some tricky points.

First, let's see how to use iii) to derive i). We can apply the standard super-dense coding protocol, and obtain a cq-channel with $d_A^2$ outputs $\{(U_m\ox \1_B) J_{AB} (U_m\ox \1_B)^\dag\}$, and the projections are given by $\{(U_m\ox \1_B)P_{AB} (U_m\ox \1_B)^\dag\}$, where $U_m$ are generalized Pauli matrices acting on $A$. So we can compute the semidefinite packing number as  
\begin{equation}\label{superdense-bound}
\frac{d_A^2}{\sum_{m=1}^{d_A^2}(U_m\ox \1_B)P_{AB} (U_m\ox \1_B)^\dag }=\frac{d_A}{\|P_{B}\|_\infty}.
\end{equation}
This is also the zero-error no-signalling assisted classical capacity of this cq-channel. 
Noticing that when $P_B<d_A\1_B$ strictly holds,  the right-hand side of the above equation 
is strictly larger than $1$.

The fact that i) implies iii) can be proven by contradiction together with the one-shot SDP formulation for $\U(K)$. Assume i) holds but $Q_B$ does not have full rank. Then we can find a non-zero vector $\ket{x}_B$ such that $Q_B\ket{x}=0$. Or equivalently, $Q_{AB}(\1_A \ox \proj{x}_B)=0$.

i) means for some $n>1$ we have $\U(K^{\ox n})>1$. By the one-shot SDP formulation of $\U(K^{\ox n})$, we can find positive semidefinite operators $S_{A^{n}}$ and $E_{A^{n}B^{n}}$, such that 
$$\tr S_{A^{n}}>1, S_{A^{n}}\ox \1_{B^{n}}\geq E_{A^{n}B^{n}}\geq 0, \tr_{A^{n}} E_{A^{n}B^{n}}=\1_{B^{n}},\ \text{and}\ \tr P_{AB}^{\ox n} (S_{A^n}\ox \1_{B^n}-E_{A^nB^n})=0.$$  
So we can find $F_{A^{n}B^{n}}=S_{A^{n}}\ox \1_{B^{n}}-E_{A^{n}B^{n}}$ with $\tr_{A^{n}} F_{A^{n}B^{n}}=(\tr S-1)\1_{B^{n}}$, with full rank. 
On the other hand, we also have $F_{A^{n}B^{n}}$ supported on $\1_{A^{n}B^{n}}-P_{AB}^{\ox n}$, 
and the later is the summation of product terms such as $P_{AB}\ox Q_{AB}\ox\cdots \ox P_{AB}$, 
containing at least  one factor $Q_{AB}$ each.  So we have $$F_{A^{n}B^{n}} (\1_{A^{n}}\ox \proj{x}^{\ox n}_{B^{n}})=0,$$
which immediately implies that $F_{B^n}=\tr_{A^n}F_{A^{n}B^{n}}$ is vanishing on the product vector $\ket{x}^{\ox n}_{B^{n}}$, contradicting the fact that $F_B$ has full rank.
\end{proof}

\medskip
The simple bound $\frac{d_A}{\|\tr_A P_{AB}\|_\infty}$ 
is very interesting, and in some important cases 
it is tight, such as  the cq-channels with symmetric outputs, and the class of Pauli 
channels. It would be interesting to know whether this kind of ``entanglement-assisted coding'' 
could provide a possible way to resolve our puzzle between $C_{0,\NS}(K)$ and $\Aram(K)$, 
eventually.

\section{Conclusion and open problems}
\label{sec:conclusion}
We have shown that there is a meaningful theory of zero-error
communication via quantum channels when assisted by quantum
no-signalling correlations.

In the terminology of non-commutative graph theory, both the one-shot 
zero-error classical capacity and simulation cost for non-commutative 
bipartite graphs assisted by quantum no-signalling correlations have 
been formulated into feasible SDPs. The asymptotic problems for non-commutative 
bipartite cq-graphs have also been successfully solved, where the capacity
turns out to involve a nontrivial regularization of super-multiplicative SDPs,
which nevertheless leads to another SDP, the semidefinite packing number.
We found analogously that the zero-error simulation cost of a cq-graph
is given by a semidefinite covering number, which in contrast to the 
classical case is in general larger than the packing number.

The zero-error classical capacity of a classical graph assisted by quantum no-signalling 
correlations is given precisely by the celebrated Lov\'asz number. 
For the most general non-commutative bipartite graphs, we are able to provide a 
necessary and sufficient condition for when these graphs have positive 
zero-error classical capacity assisted with quantum no-signalling correlations

We know rather little about the asymptotic capacity and simulation cost 
for non-commutative bipartite graphs that are not classical-quantum, however.
A very interesting candidate is the non-commutative 
bipartite graph $K(r)=\operatorname{span}\{E_0, E_1\}$ 
associated with the amplitude damping channel $\cN=\sum_{i=0}^1 E_i\cdot E_i^\dag$ ,
with $E_0=\ketbra{0}{0}+\sqrt{1-r}\ketbra{1}{1}$, $E_1=\sqrt{r}\ketbra{0}{1}$
and $0\leq r\leq 1$. 
The extremes $r=0$ and $r=1$ correspond to the noiseless qubit channel 
and a constant channel, respectively, and are trivial;
we will hence assume $0<r<1$.  This channel is quite interesting because 
it gives a separation between the semidefinite packing number and the modified version, 
and other quantities. It is also interesting because there is only one unique channel 
$\cN_r$ which has $K(r)$ as the Kraus operator space. Actually it is a non-unital 
extreme point of the convex set of CPTP maps. This channel is also able to 
communicate classical information without error when assisted with no-signalling 
correlations. By applying the above super-dense coding bound, we have
$$C_{0,\NS}(K_r)\geq \log \frac{4-2r}{3-r}.$$
By some routine calculation, we can show that the semidefinite packing number 
is $\Aram(K_r)=2-r$, the revised version $\widetilde{\Aram}(K_r)=(2-r)^2$, and
$$C_{\min {\rm E}}(K_r)=\max_{0\leq p\leq 1} H_2(p)+H_2(rp)-H_2((1-r)p),$$
where $H_2(x)=-x\log x-(1-x)\log (1-x)$ is the binary entropy.
Clearly, $C_{\min {\rm E}}(K_{0.5})=1$, while $\widetilde{\Aram}(K_{0.5})=2.25>2$. 
This means that the revised semidefinite packing number cannot be equal to the 
zero-error classical capacity assisted by quantum no-signalling correlations. 
Whether the original form has the same problem remains unknown.

There are many other interesting open problems, of which we highlight a few here.
First, it would be very 
interesting to explore the mathematical structures of quantum no-signalling 
correlations in greater detail, and to characterize quantitatively how much 
non-locality is contained in a quantum no-signalling correlation. 

Second, it is of great importance to solve the general 
asymptotic capacity and cheapest simulation 
problems in the zero-error setting when assisted with quantum no-signalling 
correlations. In particular, are there examples of strict sub-multiplicativity
of $\Sigma(K)$?

The third problem of great interest is to explore the 
relationship between the quantum Lov\'asz $\vartheta$ function introduced in 
Ref. \cite{DSW2010} for a non-commutative graph $S$, and the zero-error 
classical capacity of this non-commutative graph assisted with quantum 
no-signalling correlations, \emph{i.e.} of Kraus operator spaces $K$ with
$K^\dagger K < S$. Likewise, what is the minimum simulation cost 
$\Sigma(K)$ or $S_{0,\text{NS}}(K)$ over all such $K$? This latter is open
already for a classical graph $G$ and cq-graphs $K$ with confusability
graph $G$, see the end of section~\ref{sec:theta}. In any case, it
seems that $\Sigma(G)$ and $2^{C_{\min}(G)}$ are two new interesting, 
multiplicative graph parameters, different from both $\vartheta(G)$
and $\alpha^*(G)$.

Fourth: Are no-signalling correlations really necessary to achieve 
the asymptotic simulation cost $S_{0,\NS}(K)$, or is perhaps
entanglement enough? The motivation for this question comes from \cite{CLMW2011}
where it was shown that for classical channels, optimal asymptotic
simulation is possible using only shared randomness, no other non-local
resources are needed. A nice test case for this question is provided by
cq-graphs, for which it is easy to see that $S_{0,\NS}(K) \leq \log|B|$.
On the other hand, the entanglement-assisted simulation of cq-channels,
also known as \emph{remote state preparation} \cite{Lo,rsp} is far harder
to understand. Indeed, for generic channels the best known protocol
requires $2\log |B|$ bits of communication, which is the communication
cost of teleportation \cite{tele,Lo}.

Fifth, which no-signalling resources are actually required to achieve
the asymptotic zero-error capacity $C_{0,\NS}(K)$? This seems an
innocuous question -- after all, for each channel the SDP $\Upsilon(K)$
tells us precisely which no-signalling correlation achieves the maximum.
But looking at the case of classical channels and bipartite graphs \cite{CLMW2010},
in the light of a classical result of Elias \cite{Elias:list}, shows that
there only a very specific resource is required, which is universal for
all channels. Namely, according to \cite[Prop.~4]{Elias:list}, a list size-$L$ 
list code for a bipartite graph $\Gamma$ can achieve a rate 
$R=\left(1-\frac1L\right)\log\alpha^*(\Gamma) - O\left(\frac1L\right)$, which
is arbitrarily close to $\log\alpha^*(\Gamma)$ for sufficiently large $L$.
Using such a code, Alice can send a message $i$ out of $M=\lfloor 2^{nR} \rfloor$
over $n$ uses of the channel, and Bob will end up with a list, i.e.~a subset
$I \subset 2^{[M]}$ of $|I|=L$ possible messages such that $i\in I$.
To resolve the remaining ambiguity, Alice and Bob now require a no-signalling
resource $S_{{M\choose L}} = S(\alpha\beta|iI)$ that we call \emph{subset correlation},
defined for $i\in [M]$, $I\in {[M]\choose L}$, $t,u\in[L]$, with
\[
  S(tu|iI) = \begin{cases}
               \frac{1}{L} & \text{ if } i = i_s\in I=\{i_1 < \ldots < i_L\}\ \&{}\ u-t   = s \mod L, \\
               \frac{1}{L} & \text{ if } i = i_s\in I=\{i_1 < \ldots < i_L\}\ \&{}\ u-t\neq s \mod L, \\
               \frac{1}{L^2} & \text{ if } i\not\in I.
             \end{cases}
\]
Alice will input $i$ into the box, Bob the set $I$, and then Alice will send
her output $t$ to Bob; since $i\in I$, Bob knows that he can recover the
index $s$ of $i=i_s \in  I=\{i_1 < \ldots < i_L\}$ as $u-t \mod L$.
Do these or other universal quantum no-signalling correlations allow
to achieve the $C_{0,\NS}(K)$ (or at least $\log\Aram(K)$ for cq-graphs)?

Finally, in \cite{LaiDuan-theta}, the observation was made 
that the optimal orthogonal representation of a graph $G$ in the sense of
Lov\'{a}sz (see Theorem~\ref{thm:theta}) seems to satisfy $\Upsilon(K)=\vartheta(G)$,
which was confirmed by direct calculation for several graph families. If this
held in general, it would mean that $\vartheta(G)$ can in a certain sense
be achieved by a single channel use, rather than requiring a many-copy limit.

\acknowledgments
It is a pleasure to thank many people for their interests and feedback on the 
present work, including Salman Beigi, Mick Bremner, Eric Chitambar, 
Min-Hsiu Hsieh, Aram Harrow, Masahito Hayashi, Richard Jozsa, Ching-Yi Lai, Debbie Leung, 
Will Matthews, Jonathan Oppenheim, Xin Wang, and Stephanie Wehner. Further thanks are due
to Giulio Chiribella, Marco Piani, and John Watrous for delightful discussions
on quantum no-signalling correlations. 

RD was supported in part by the Australian Research Council (ARC) under 
Grant DP120103776 (with AW) and by the National Natural Science Foundation 
of China under grants no.~61179030. He was also supported in part by an ARC 
Future Fellowship under Grant FT120100449.
AW was supported by the European Commission (STREPs ``QCS''
and ``RAQUEL''), the European Research Council (Advanced Grant ``IRQUAT'') and 
the Philip Leverhulme Trust.
Furthermore, by the Spanish MINECO, projects FIS2008-01236 and
FIS2013-40627-P, with the support of FEDER funds, as well as
the Generalitat de Catalunya CIRIT, project no.~2014 SGR 966.

\end{document}